%% file: xkummer.tex
\documentclass[plain,11pt]{llncs}
\usepackage{etex}
\usepackage[utf8]{inputenc}
\usepackage[T1]{fontenc}
\usepackage{ctable}
\usepackage{multirow}
\usepackage{setspace}
\usepackage{enumitem}
\usepackage{a4wide}

\newif\ifpublic
\publictrue

\newif\ifapp
\apptrue

\usepackage{amsfonts,amsmath,amscd,amssymb,array}
\usepackage{mathtools}

\usepackage{type1cm}

\usepackage{fancyhdr}
\usepackage[ruled,vlined,linesnumbered]{algorithm2e}
\SetEndCharOfAlgoLine{}

\usepackage{xcolor}
\definecolor{linkcolor}{rgb}{0.65,0,0}
\definecolor{citecolor}{rgb}{0,0.65,0}
\definecolor{urlcolor}{rgb}{0,0,0.65}
\usepackage[colorlinks=true, backref=page, linkcolor=linkcolor, urlcolor=urlcolor, citecolor=citecolor]{hyperref}
\usepackage{breakurl}

\usepackage{cite}
\usepackage{xspace}

\usepackage[all]{xy}

\renewcommand{\tabcolsep}{4pt}

\input{macros}

\begin{document}

\titlerunning{qDSA: Small and Secure Digital Signatures\\
    with Curve-based Diffie--Hellman Key Pairs
}
\title{
	qDSA: Small and Secure Digital Signatures 
    \\
    with Curve-based Diffie--Hellman Key Pairs
}

\ifpublic
\author{
  Joost~Renes\inst{1}\thanks{
This work has been supported 
by the Technology Foundation STW (project 13499 - TYPHOON \& ASPASIA), from the Dutch government.
}
\and
  Benjamin~Smith\inst{2} 
}
\institute{
  Digital Security Group, Radboud University, The Netherlands\\
  \email{j.renes@cs.ru.nl}
  \and
  INRIA \emph{and} Laboratoire d'Informatique de l'École polytechnique
  (LIX),\\
  Université Paris--Saclay, France\\
  \email{smith@lix.polytechnique.fr}
}
\fi 

\maketitle
\input{abstract}

\input{intro}
\input{signatures}

\input{elliptic}

\input{kummer}

\input{verification}
\input{compression}

\input{implementation}

\ifpublic
\subsubsection{Acknowledgements.}
We thank Peter Schwabe for his valuable contributions
to discussions during the creation of this paper,
and the anonymous Asiacrypt reviewers for their
helpful comments.
\fi


\bibliographystyle{abbrv}
\bibliography{bib}

\ifapp
\appendix
\input{elliptic-impl}
\input{kummer-impl}

\fi

\end{document}

%% file: macros.tex
\def\F{\mathbb{F}}
\def\Z{\mathbb{Z}}

\def\arraystretch{1.2}

\newcommand{\concat}{\mid\mid}

\newcommand{\Pc}{P}
\newcommand{\Q}{Q}
\newcommand{\R}{R}

\newcommand{\Sc}{\mathcal{S}}
\newcommand{\T}{\mathcal{T}}
\newcommand{\J}{\mathcal{J}}
\newcommand{\K}{\mathcal{K}}
\newcommand{\C}{\mathcal{C}}
\newcommand{\Ocal}{\mathcal{O}}

\def\armjac2scal{$4\,270\,563$\xspace}

\def\avrjac2scal{$15\,701\,276$\xspace}

\newcommand{\ie}{i.\,e.\ }
\newcommand{\eg}{eg.\xspace}

\newcommand{\Jac}{\ensuremath{\mathcal{J}}}
\newcommand{\Kum}{\ensuremath{\mathcal{K}}}
\newcommand{\KChudnovsky}{\ensuremath{\Kum^{\mathrm{Sqr}}}}
\newcommand{\KGaudry}{\ensuremath{\Kum^{\mathrm{Can}}}}
\newcommand{\KInter}{\ensuremath{\Kum^{\mathrm{Int}}}}
\newcommand{\KChudnovskydual}{\ensuremath{\dualof{\Kum}^{\mathrm{Sqr}}}}
\newcommand{\KGaudrydual}{\ensuremath{\dualof{\Kum}^{\mathrm{Can}}}}
\newcommand{\KInterdual}{\ensuremath{\dualof{\Kum}^{\mathrm{Int}}}}
\newcommand{\KTetra}{\ensuremath{\Kum^{\mathrm{Tet}}}}
\newcommand{\BCan}{\ensuremath{B^{\mathrm{Can}}}}
\newcommand{\BCandual}{\ensuremath{\dualof{B}^{\mathrm{Can}}}}
\newcommand{\BSqr}{\ensuremath{B^{\mathrm{Sqr}}}}
\newcommand{\BInt}{\ensuremath{B^{\mathrm{Int}}}}

\newcommand{\subgrp}[1]{{\left\langle{#1}\right\rangle}}
\newcommand{\ZZ}{\Z}
\newcommand{\FF}{\mathbb{F}}
\newcommand{\PP}{\mathbb{P}}

\newcommand{\PPsqr}{{\texttt{Sqr4}}} 
\newcommand{\PPmul}{{\texttt{Mul4}}} 
\newcommand{\Hadamard}{{\texttt{Had}}} 
\newcommand{\Dotprod}{\texttt{Dot}}

\newcommand{\field}{\ensuremath{\Bbbk}}

\newcommand{\xADD}{\texttt{xADD}\xspace}
\newcommand{\xDBL}{\texttt{xDBL}\xspace}
\newcommand{\xDBLADD}{\texttt{xDBLADD}\xspace}

\newcommand{\dualof}[1]{\ensuremath{\widehat{#1}}}

\newcommand{\Hadamardn}{\ensuremath{\mathcal{H}}}
\newcommand{\Hadamardd}{\dualof{\Hadamardn}}
\newcommand{\Squaren}{\ensuremath{\mathcal{S}}}
\newcommand{\Squared}{\dualof{\Squaren}}
\newcommand{\Scalen}{\ensuremath{\mathcal{C}}}
\newcommand{\Scaled}{\dualof{\Scalen}}
\newcommand{\Tetran}{\ensuremath{\mathcal{T}}}

\newcommand{\XGn}{T}
\newcommand{\XCn}{X}
\newcommand{\XHn}{Y}
\newcommand{\XGd}{\dualof{\XGn}}
\newcommand{\XCd}{\dualof{\XCn}}
\newcommand{\XHd}{\dualof{\XHn}}
\newcommand{\XTn}{L}

\newcommand{\unsquaredn}{\alpha}
\newcommand{\squaredn}{\mu}
\newcommand{\unsquaredd}{\dualof{\alpha}}
\newcommand{\squaredd}{\dualof{\mu}}
\newcommand{\epsilond}{\dualof{\epsilon}}
\newcommand{\kappad}{\dualof{\kappa}}
\newcommand{\an}{\unsquaredn_1}
\newcommand{\bn}{\unsquaredn_2}
\newcommand{\cn}{\unsquaredn_3}
\newcommand{\dn}{\unsquaredn_4}
\newcommand{\aan}{\squaredn_1}
\newcommand{\bbn}{\squaredn_2}
\newcommand{\ccn}{\squaredn_3}
\newcommand{\ddn}{\squaredn_4}
\newcommand{\ad}{\unsquaredd_1}
\newcommand{\bd}{\unsquaredd_2}
\newcommand{\cd}{\unsquaredd_3}
\newcommand{\dd}{\unsquaredd_4}
\newcommand{\aad}{\squaredd_1}
\newcommand{\bbd}{\squaredd_2}
\newcommand{\ccd}{\squaredd_3}
\newcommand{\ddd}{\squaredd_4}
\newcommand{\En}{E}
\newcommand{\Fn}{F}
\newcommand{\Gn}{G}
\newcommand{\Hn}{H}
\newcommand{\Ed}{\dualof{\En}}
\newcommand{\Fd}{\dualof{\Fn}}
\newcommand{\Gd}{\dualof{\Gn}}
\newcommand{\Hd}{\dualof{\Hn}}

\newcommand{\zplus}{\Z_N^+}
\newcommand{\xid}{\texttt{qID}\xspace}
\newcommand{\xsig}{\texttt{qSIG}\xspace}
\newcommand{\xsign}{\texttt{qDSA}\xspace}

\newcommand{\hashplus}[1]{\overline{#1}}

\newcommand{\prover}{\texttt{prover}\xspace}
\newcommand{\verifier}{\texttt{verifier}\xspace}

\newcommand{\MUL}{{\bf M}}
\newcommand{\SQR}{{\bf S}}
\newcommand{\MLC}{{\bf C}}
\newcommand{\ADD}{{\bf a}}
\newcommand{\SUB}{{\bf s}}
\newcommand{\INV}{{\bf I}}
\newcommand{\SQRT}{{\bf E}}

%% file: abstract.tex
\begin{abstract}
    qDSA is a high-speed, high-security signature scheme
    that facilitates implementations with a very small memory footprint,
    a crucial requirement for embedded systems and IoT devices,
    and that
    uses the same public keys as modern Diffie--Hellman schemes
    based on Montgomery curves (such as Curve25519) or Kummer surfaces.
    qDSA resembles an adaptation of EdDSA
    to the world of Kummer varieties,
    which are quotients of algebraic groups by \(\pm1\).
    Interestingly, qDSA does not require any full group operations
    or point recovery: all computations,
    including signature verification, 
    occur on the quotient where there is no group law.
    We include details on four implementations of qDSA,
    using Montgomery and fast Kummer surface arithmetic
    on the 8-bit AVR {ATmega} and 32-bit ARM Cortex~M0 platforms.
	We find that qDSA significantly outperforms state-of-the-art
    signature implementations in terms of stack usage and code size.
    We also include an efficient compression algorithm
    for points on fast Kummer surfaces,
    reducing them to the same size as compressed elliptic curve points for
    the same security level.

    \noindent \textbf{Keywords.} 
    Signatures, Kummer, Curve25519, Diffie--Hellman, elliptic curve,
    hyperelliptic curve.
\end{abstract}

%% file: intro.tex
\section{Introduction}
\label{sec:intro}

Modern asymmetric cryptography based on elliptic and hyperelliptic curves~\cite{koblitz,miller}
achieves two important goals.
The first is efficient key exchange 
using the 
Diffie--Hellman protocol~\cite{diffie-hellman},
using the fact that 
the (Jacobian of the) curve carries the structure of an abelian group.
But in fact, as Miller observed~\cite{miller}, 
we do not need the full group structure for Diffie--Hellman:
the associated \emph{Kummer variety} (the quotient by \(\pm1\)) suffices,
which permits
more efficiently-computable arithmetic~\cite{M87,Gaudry}.
Perhaps the most well-known example is Curve25519~\cite{Ber06},
which offers fast scalar multiplications based on \(x\)-only arithmetic.

The second objective is efficient digital signatures, 
which are critical for authentication.
There are several group-based signature schemes, 
the most important
of which are ECDSA~\cite{ECDSA}, 
Schnorr~\cite{schnorr},
and now EdDSA~\cite{EdDSA} signatures.
In contrast to the Diffie--Hellman protocol, 
all of these signature schemes explicitly require
the group structure of the (Jacobian of the) curve.
An unfortunate side-effect of this is that 
users essentially need two public keys
to support both curve-based protocols.
Further,
basic cryptographic libraries need to provide implementations for 
arithmetic on both the Jacobian and the Kummer variety,
thus complicating and increasing the size of the trusted code base.
For example, the NaCl library~\cite{nacl} 
uses Ed25519~\cite{EdDSA} for signatures, 
and Curve25519~\cite{Ber06} for key exchange.
This problem is worse for genus-2 hyperelliptic curves, where
the Jacobian is significantly harder to use safely than its Kummer surface.

There have been several partial solutions to this problem.
By observing that elements of the Kummer variety are elements of the Jacobian \emph{up to sign},
one can build scalar multiplication on the Jacobian based on the fast Kummer
arithmetic~\cite{Okeya--Sakurai,Recovery}. 
This avoids the need for a separate scalar multiplication on the Jacobian, 
but does not avoid the need for its group law;
it also introduces the need for projecting to and recovering from the
Kummer.
In any case, it does not solve the problem of having different public key types.

Another proposal is XEdDSA~\cite{xeddsa},
which uses the public key on the Kummer variety to construct EdDSA signatures.
In essence, it creates a key pair on the Jacobian
by appending a sign bit to the public key on the Kummer variety,
which can then be used for signatures.
In~\cite{H12} Hamburg shows that one can actually verify signatures using only the
\(x\)-coordinates of points on an elliptic curve,
	which is applied in the
recent STROBE framework~\cite{H17}. 
We generalize this approach to allow Kummer varieties of 
curves of higher genera, and naturally adapt the scheme
by only allowing challenges
up to sign.
This allows us to provide a proof of security, which
has thus far not been attempted
(in~\cite{H12} Hamburg remarks that verifying up to sign
does ``probably not impact security at all'').
Similar techniques have been applied 
for batch verification of ECDSA signatures~\cite{KD14}, 
using the theory of summation polynomials~\cite{summation}.

In this paper we show that there is no intrinsic reason why 
Kummer varieties cannot be used for signatures.
We present \xsign,
a signature scheme relying only on Kummer arithmetic,
and prove it secure in the random oracle model.
It should not be surprising that the reduction in our
proof is slightly weaker than the standard proof
of security of Schnorr signatures~\cite{PS96}, 
but not by more than we should expect.
There is no difference between public keys for \xsign and Diffie--Hellman.
After an abstract presentation in~\S\ref{sec:signatures},
we give a detailed description of 
elliptic-curve \xsign instances in~\S\ref{sec:elliptic}.
We then move on to genus-2 instances based on fast Kummer surfaces,
which give better performance.
The necessary arithmetic appears in~\S\ref{sec:kummer-arithmetic},
before~\S\ref{sec:verif} describes the new verification algorithm.

We also provide an efficient compression method
for points on fast Kummer surfaces in~\S\ref{sec:comp}, 
solving a long-standing open problem~\cite{B-ecc06-talk}.
Our technique means that \xsign public keys for \(g=2\) 
can be efficiently compressed to 32 bytes,
and that \xsign signatures fit into 64 bytes;
it also finally reduces the size of 
Kummer-based Diffie--Hellman public keys
from 48 to 32 bytes.

Finally, we provide constant-time software implementations 
of genus-1 and genus-2 \xsign instances 
for the AVR ATmega and ARM Cortex M0 platforms.
The performance of all four \xsign implementations,
reported in~\S\ref{sec:imp},
comfortably beats earlier implementations 
in terms of stack usage and code size.

\paragraph{Source code.}

We place all of the software described here
into the public domain, to maximize the reusability of our results.
\ifpublic
The software is available at \url{http://www.cs.ru.nl/~jrenes/}.
\fi

%% file: signatures.tex
\section{The \xsign signature scheme}\label{sec:signatures}

In this section we define \xsign,
the
\emph{quotient Digital Signature Algorithm}.
We start by 
recalling the basics of Kummer varieties in~\S\ref{subsec:kummervar} 
and defining key operations in~\S\ref{sec:basic-ops}.
The rest of the section is dedicated to the definition of the
\xsign signature scheme, which is presented in full in 
Algorithm~\ref{alg:xsign}, and its proof of security,
which follows Pointcheval and Stern~\cite{PS96,PS00}.
\xsign closely resembles the Schnorr signature scheme~\cite{schnorr},
as it results from applying the Fiat--Shamir heuristic~\cite{FS86}
to an altered Schnorr identification protocol,
together with a few standard changes as in EdDSA~\cite{EdDSA}.
We comment on some special properties of \xsign in \S\ref{subsec:xsign}.

Throughout, we work over finite fields \(\FF_p\) with \(p > 3\).
Our low-level algorithms include costs in terms of basic
\(\FF_p\)-operations:
\(\MUL\), \(\SQR\), \(\MLC\), \(\ADD\), \(\SUB\), \(\INV\), and \(\SQRT\)
denote
the unit costs of computing a single
multiplication, squaring, multiplication by a small
constant, addition, subtraction, inverse, and square root, respectively.

\subsection{The Kummer variety setting}\label{subsec:kummervar}

Let \(\C\) be a (hyper)elliptic curve and \(\J\) its Jacobian\footnote{%
    In what follows, we could replace \(\J\) by an arbitrary 
    abelian group and
    all the proofs would be completely analogous. 
    For simplicity we restrict to the cryptographically most interesting case
    of a Jacobian.
}.
The Jacobian is a commutative algebraic group with 
group operation \(+\), inverse \(-\), and identity \(0\).
We assume \(\J\) has a subgroup of large prime order \(N\).
The associated \emph{Kummer variety} \(\K\) is the quotient \(\K=\J/\pm\).
By definition, 
working with \(\K\) corresponds to working on \(\J\) \emph{up to sign}.
If \(\Pc\) is an element of \(\J\), we denote its image in \(\K\)
by \(\pm\Pc\).
In this paper we take \(\log_2 N\approx 256\), and consider two important cases.
\begin{description}
	\item[Genus 1.] 
        Here \(\J = \C/\FF_p\) is an elliptic curve 
        with \(\log_2 p\approx 256\),
        while \(\K=\PP^1\) is 
        the \(x\)-line. We choose \(\C\) to be Curve25519~\cite{Ber06}, 
        which is the topic of \S\ref{sec:elliptic}.
	\item[Genus 2.] 
        Here \(\J\) is the Jacobian of a genus-2 curve \(\C/\FF_p\),
        where \(\log_2 p\approx 128\), 
        and \(\K\) is a \emph{Kummer surface}. 
        We use the Gaudry--Schost parameters~\cite{gaudry-schost} 
        for our implementations.
        Kummer arithmetic, including some new constructions we need for 
        signature verification and compression, is described
        in \S\ref{sec:kummer-arithmetic}-\ref{sec:comp}.
\end{description}

A point \(\pm P\) in \(\K(\FF_p)\)
is the image of a pair of points \(\{P,-P\}\) on \(\Jac\).
It is important to note that 
\(P\) and \(-P\) are not necessarily in \(\Jac(\FF_p)\);
if not, then they are conjugate points in \(\Jac(\FF_{p^2})\),
and correspond to points in \(\Jac'(\FF_p)\),
where \(\Jac'\) is the \emph{quadratic twist} of \(\Jac\).
Both \(\Jac\) and \(\Jac'\) always have the same Kummer variety;
we return to this fact, and its implications for our scheme,
in~\S\ref{subsec:xsign} below.


\subsection{Basic operations}
\label{sec:basic-ops}

While a Kummer variety \(\K\) has no group law,
the operation
\begin{equation}
    \label{eq:pseudo-add}
    \left\{\pm\Pc, \pm \Q\right\} 
    \mapsto
    \left\{\pm(\Pc+\Q),\pm(\Pc-\Q)\right\}
\end{equation}
is well-defined.
We can therefore define a \emph{pseudo-addition}
operation by
\[
    \xADD: \left(\pm\Pc, \pm\Q, \pm (\Pc-\Q)\right)
    \mapsto
    \pm (\Pc+\Q).
\]
The special case where \(\pm (\Pc - \Q) = \pm 0\)
is the \emph{pseudo-doubling}
\(
    \xDBL: \pm \Pc \mapsto \pm [2]\Pc
\).
In our applications we 
can often improve efficiency by 
combining two of these operations in a single function
\[
    \xDBLADD: \left(\pm\Pc, \pm\Q, \pm (\Pc-\Q)\right)
    \longmapsto
    \left(\pm[2]\Pc,\pm (\Pc+\Q)\right)
    \,.
\]
For any integer \(m\), the scalar multiplication \([m]\)
on \(\J\) induces the key cryptographic operation of
\emph{pseudomultiplication} on \(\K\), 
defined by
\[
    {\tt Ladder} : (m,\pm\Pc) \longmapsto \pm[m]\Pc
    \,.
\]
As its name suggests,
we compute {\tt Ladder} using 
Montgomery's famous ladder algorithm~\cite{M87},
which is a uniform sequence of \(\xDBLADD\)s
and constant-time conditional swaps.\footnote{%
    In contemporary implementations such as NaCl, 
    the {\tt Ladder} function is sometimes named \texttt{crypto\_scalarmult}.
}
This constant-time nature will be important for signing.

Our signature verification requires a function
\( {\tt Check} \) on \(\K^3 \)
defined by
\[
    {\tt Check}
    :
    (\pm\Pc, \pm\Q, \pm\R) 
    \longmapsto
    \begin{cases}
        {\bf True} &
        \text{if } \pm\R\in\left\{\pm(\Pc+\Q),\pm(\Pc-\Q)\right\}
        \\
        {\bf False} &
        \text{otherwise}
    \end{cases}
\]

Since we are working with projective points, 
we need a way to uniquely represent them.
Moreover, we want this representation to be as small as possible, 
to minimize communication overhead.
For this purpose we define the functions
\[
    {\tt Compress} : \K(\FF_p) \longrightarrow \{0,1\}^{256}
    \,,
\]
writing \(\overline{\pm\Pc}:={\tt Compress}(\pm\Pc)\),
and 
\[
    {\tt Decompress} : 
    \{0,1\}^{256} 
    \longrightarrow 
    \K(\FF_p) \cup \{\bot\}
\]
such that 
\({\tt Decompress}(\overline{\pm\Pc}) = \pm\Pc\)
for \(\pm P\) in \(\Kum(\FF_p)\)
and
\({\tt Decompress}(X)=\bot\)
for \(X\in\{0,1\}^{256}\setminus\operatorname{Im}({\tt Compress})\).

For the remainder of this section we assume 
that \texttt{Ladder},
\texttt{Check},
\texttt{Compress},
and
\texttt{Decompress}
are defined.
Their implementation depends on whether we are in the genus 1 or 2
setting; we return to this in later sections.

\subsection{The \xid identification protocol}
\label{subsec:id}

Let \(\Pc\) be a generator of a prime order subgroup of \(\J\),
of order \(N\), and let \(\pm\Pc\) be its image in \(\K\). 
Let \(\zplus\) denote the subset of \(\Z_N\)
with zero least significant bit
(where we identify elements of \(\Z_N\) 
with their representatives in \([0,N-1]\)).
Note that since \(N\) is odd, \({\tt LSB}(-x)=1-{\tt LSB}(x)\) 
for all \(x\in\Z_N^*\).
The private key is an element \(d\in\Z_N\).
Let \(\Q=[d]\Pc\) and let the public key be \(\pm\Q\).
Now consider the following Schnorr-style identification protocol,
which we call \xid:
\begin{enumerate}
	\item The \prover sets \(r\gets_R\Z_N^*\), \(\pm\R\gets\pm[r]\Pc\)
	and sends \(\pm\R\) to the \verifier;
	\item The \verifier sets \(c\gets_R\Z_N^+\) and sends \(c\) to the \prover;
	\item The \prover sets \(s\gets(r-cd)\mod{N}\) and sends \(s\) to the \verifier;
	\item The \verifier accepts if and only if 
	\(\pm\R\in\left\{\pm([s]\Pc+[c]\Q),\pm([s]\Pc-[c]\Q)\right\}\).
\end{enumerate}
There are some
important differences between \xid 
and the basic Schnorr identification protocol in~\cite{schnorr}.
\begin{description}
	\item[Scalar multiplications on \(\K\).] 
		It is well-known that one can use \(\K\) to perform the scalar
		multiplication~\cite{Okeya--Sakurai,Recovery} 
        within a Schnorr
		identification or signature scheme,
		but
        with this approach one must always lift back to an element of a group.
		In contrast, in our scheme this recovery step is not necessary.
	\item[Verification on \(\K\).]
		The original verification~\cite{schnorr} 
		requires checking that \(\R = [s]\Pc + [c]\Q\)
		for some \(\R, [s]\Pc, [c]\Q\in\J\).
		Working on \(\K\), we only have these values up to sign 
		(\ie \(\pm\R\), \(\pm[s]\Pc\) and \(\pm[c]\Q\)), 
		which is not enough to check that \(\R=[s]\Pc+[c]\Q\).
		Instead, we only verify that \(\pm\R=\pm\left([s]\Pc\pm[c]\Q\right)\).
	\item[Challenge from \(\zplus\).]
        A Schnorr protocol using the weaker verification above
		would not satisfy the special soundness property:
		the transcripts 
		\(\left(\pm\R,c,s\right)\)
		and
		\(\left(\pm\R,-c,s\right)\)
		are both valid, and do not allow us to extract a witness.
        Choosing \(c\) from \(\zplus\) instead of \(\ZZ\) eliminates
        this possibility,  and allows a security proof 
        (this is the main difference with Hamburg's STROBE~\cite{H17}).
\end{description}

\begin{proposition}\label{thm:sigma}
    The \xid identification protocol is a sigma protocol.
\end{proposition}
\begin{proof}
    We prove the required properties (see~\cite[\S6]{hazaylindell}). 
    
    \textbf{Completeness:} 
    If the protocol is followed, 
    then \(r=s+cd\), 
    and therefore
    \([r]\Pc=[s]\Pc+[c]\Q\) on \(\J\).
    Mapping to \(\K\), it follows that \(\pm\R=\pm([s]\Pc+[c]\Q)\).

    \textbf{Special soundness:}
    Let \(\left(\pm\R,c_0,s_0\right)\) 
    and \(\left(\pm\R,c_1,s_1\right)\) be two valid transcripts
    such that \(c_0\neq c_1\). 
    By verification,
    each \(s_i \equiv \pm r\pm c_id \pmod{N}\),
    so \(s_0\pm s_1 \equiv \left(c_0\pm c_1\right)d \pmod{N}\), 
    where the signs are chosen to cancel \(r\).
    Now
    \(c_0\pm c_1\not\equiv0\pmod{N}\)
    because
    \(c_0\) and \(c_1\) are both in \(\ZZ_N^+\),
    so we can extract a witness 
    \(d\equiv\left(s_0\pm s_1\right)\left(c_0\pm c_1\right)^{-1}\pmod{N}\).

    \textbf{Honest-verifier zero-knowledge:}
    A simulator \(\Sc\) generates \(c\leftarrow_R\Z_N^+\)
    and sets \(s\leftarrow_R\Z_N\)
    and \(\R\leftarrow [s]\Pc+[c]\Q\).%
    \footnote{As we only know \(\Q\) up to sign, we may need
    two attempts to construct \(\Sc\).}
    If \(\R=\Ocal\), it restarts.
    It outputs
    \(\left(\pm\R,c,s\right)\). 
    As in~\cite[Lemma~5]{PS00}, 
    we let 
    \begin{align*}
        \delta 
        & =
        \left\{
            \left(\pm\R, c, s\right)
            :
                c \in_R \Z_N^+ \,,
                r \in_R \Z_N^* \,,
                \pm\R = \pm[r]\Pc \,,
                s = r-cd
        \right\}
        \,,
        \\
        \delta' 
        & =
        \left\{
            \left(\pm\R, c, s\right)
            :
                c \in_R \Z_N^+ \,,
                s \in_R \Z_N \,,
                \R = [s]\Pc+[c]\Q \,,
                \R \neq \Ocal
        \right\}
    \end{align*}
    be the distributions
    of honest and simulated signatures, respectively.
    The elements of \(\delta\) and \(\delta'\) are the same. 
    First, consider \(\delta\). 
    There are exactly \(N-1\) choices for \(r\), 
    and exactly \((N+1)/2\) for \(c\);
    all of them lead to distinct tuples.
    There are thus \((N^2-1)/2\) possible tuples, 
    all of which have probability \(2/(N^2-1)\) of occurring.
    Now consider \(\delta'\). Again,
    there are \((N+1)/2\) choices for \(c\).
    We have \(N\) choices for \(s\), 
    exactly one of which leads to \(\R=\Ocal\).
    Thus, given \(c\), 
    there are \(N-1\) choices for \(s\). 
    We conclude that \(\delta'\) also contains \((N^2-1)/2\) possible tuples, 
    which all have probability \(2/(N^2-1)\) of occurring.
    \qed
\end{proof}

\subsection{Applying Fiat--Shamir}\label{subsec:schnorr}

Applying the Fiat--Shamir transform~\cite{FS86} to \xid
yields
a signature scheme \xsig.
We will need a hash function
\(\hashplus{H}:\{0,1\}^* \to \zplus\),
which we define by taking a hash function 
\(H:\{0,1\}^* \to \Z_N\) and then setting \(\hashplus{H}\) 
by
\[
    \hashplus{H}(M) 
    \longmapsto
    \begin{cases}
        H(M)  &\text{if } {\tt LSB}(H(M)) = 0
        \\
        -H(M) &\text{if } {\tt LSB}(H(M)) = 1
    \end{cases}.
\]
The \xsig signature scheme is defined as follows:
\begin{enumerate}
	\item 
        To sign a message \(M\in\{0,1\}^*\) with private key \(d\in\Z_N\) 
	    and public key \(\pm\Q\in\K\), 
	    the \prover sets 
        \( r \gets_R \Z_N^* \), 
        \(\pm\R \gets \pm[r]\R\), 
        \(h \gets \hashplus{H}(\pm\R\concat M)\),
        and \(s \gets (r-hd)\bmod{N}\),
	    and sends \(\left(\pm\R\concat s\right)\) to the \verifier.
	\item 
        To verify a signature \(\left(\pm\R\concat s\right)\in\K\times\Z_N\) 
        on a message \(M\in\{0,1\}^*\)
        with public key \(\pm\Q\in\K\),
        the \verifier sets 
        \( h \gets \hashplus{H}(\pm\R\concat M)\),
        \( \pm\T_0 \gets \pm[s]\Pc\),
        and \( \pm\T_1 \gets \pm[h]\Q\),
        and accepts if and only if 
        \(\pm\R\in\left\{\pm(\T_0+\T_1),\pm(\T_0-\T_1)\right\}\).
\end{enumerate}
Proposition~\ref{thm:schfs}
asserts that
the security properties of \xid
carry over to \xsig.

\begin{proposition}\label{thm:schfs}
    In the random oracle model, if an existential forgery of the
    \xsig signature scheme under an adaptive chosen message attack
    has non-negligible probability of success, 
    then the DLP in \(\J\) can be solved in polynomial time.
\end{proposition}
\begin{proof}
    This is the standard proof of applying the Fiat--Shamir transform
    to a sigma protocol:
    see~\cite[Theorem~13]{PS96} or~\cite[\S3.2]{PS00}.
    \qed
\end{proof}

\subsection{The \xsign signature scheme}
\label{subsec:xsign}

Moving towards the real world,
we slightly alter the \xsig protocol with some pragmatic choices,
following Bernstein et al.~\cite{EdDSA}:
\begin{enumerate}
    \item
        We replace the randomness \(r\) by the output
        of a pseudo-random function,
        which makes the signatures deterministic.
    \item
        We include the public key \(\pm\Q\) in the generation of the challenge,
        to prevent attackers from 
        attacking multiple public keys at the same time.
    \item
        We compress and decompress points on \(\K\) where necessary.
\end{enumerate}
The resulting signature scheme, \xsign, 
is summarized in Algorithm~\ref{alg:xsign}.

\begin{algorithm}
    \setstretch{1.2} 
    \caption{The \xsign signature scheme}
    \label{alg:xsign}
    \SetKwProg{Function}{function}{}{}
    \Function{{\tt keypair}}{
        \KwIn{()}
        \KwOut{%
            \((\overline{\pm\Q}\concat\left(d'\concat d''\right))\): 
            a compressed public key \(\overline{\pm \Q}\in\{0,1\}^{256}\)
            where \(\pm \Q\in \Kum\),
            and a private key \(\left(d'\concat d''\right)\in\left(\{0,1\}^{256}\right)^2\)
        }
        \( d \gets \texttt{Random}(\{0,1\}^{256}) \)
        \;
        \( \left(d'\concat d''\right) \gets H(d) \)\label{line:dprime}
        \;
        \(\pm \Q \gets \texttt{Ladder}(d',\pm \Pc) \)
        \tcp*{\(\pm \Q = \pm[d']\Pc\)}
        \(\overline{\pm \Q} \gets \texttt{Compress}(\pm \Q)\)
        \;
        \Return{\((\overline{\pm \Q}\concat\left(d'\concat d''\right))\)}
    }
    \Function{{\tt sign}}{
        \KwIn{%
            \(d',d''\in\{0,1\}^{256}\),
            \(\overline{\pm \Q}\in\{0,1\}^{256}\),
            \(M\in\{0,1\}^*\)
        }
        \KwOut{
            \((\overline{\pm \R}\concat s)\in\left(\{0,1\}^{256}\right)^2\)
        }
        \(r \gets H(d''\concat M)\)\label{line:r}
        \;
        \(\pm \R \gets \texttt{Ladder}(r,\pm \Pc)\)
        \tcp*{\(\pm \R = \pm[r]\Pc\)}
        \(\overline{\pm \R} \gets \texttt{Compress}(\pm \R)\)
        \;
        \(h \gets \hashplus{H}(\overline{\pm \R} \concat \overline{\pm\Q} \concat M)\)
        \label{line:h}
        \;
        \(s \gets (r - hd') \bmod{N}\)
        \;
        \Return{\((\overline{\pm \R}\concat s)\)}
    }
    \Function{{\tt verify}}{
        \KwIn{
            \(M \in \{0,1\}^*\),
            the compressed public key \(\overline{\pm \Q} \in \{0,1\}^{256}\),
            and 
            a putative signature 
            \((\overline{\pm \R}\concat s) \in \left(\{0,1\}^{256}\right)^2\) 
        }
        \KwOut{%
            \texttt{True} if \((\overline{\pm \R}\concat s)\) is a valid signature
            on \(M\) under \(\overline{\pm \Q}\),
            \texttt{False} otherwise
        }
        \(\pm \Q \gets \texttt{Decompress}(\overline{\pm \Q})\)
        \;
        \If{\(\pm \Q = \bot\)}{
            \Return{{\rm\tt False}}
        }
        \(h \gets \hashplus{H}(\overline{\pm \R} \concat \overline{\pm\Q}\concat M)\)
        \;
        \(\pm \T_0 \gets \texttt{Ladder}(s,\pm \Pc)\)
        \tcp*{\(\pm \T_0 = \pm[s]\Pc\)}
        \(\pm \T_1 \gets \texttt{Ladder}(h,\pm \Q)\)
        \tcp*{\(\pm \T_1 = \pm[h]\Q\)}
        \(\pm \R \gets \texttt{Decompress}(\overline{\pm \R})\)
        \;
        \If{\(\pm \R = \bot\)}{%
            \Return{{\rm\tt False}}
        }
        \(v \gets \texttt{Check}(\pm \T_0, \pm \T_1, \pm \R)\)
        \tcp*{is \(\pm \R = \pm \left(\T_0 \pm \T_1\right)\)?}
        \Return{\(v\)}
    }
\end{algorithm}

\paragraph{Unified keys.} 
Signatures are entirely computed and verified on \(\K\),
which is also the natural setting for Diffie--Hellman key exchange.
We can therefore use identical key pairs for Diffie--Hellman
and for \xsign signatures.
This significantly simplifies the implementation of cryptographic
libraries, as we no longer
need arithmetic for the two distinct objects \(\J\) and \(\K\).
Technically, there is no reason not to use a single
key pair for both key exchange and signing;
but one should be very careful in doing so, 
as using one key across multiple protocols 
could potentially lead to attacks.
The primary interest of this aspect of \xsign 
is not necessarily in reducing the number of keys,
but in unifying key formats and reducing the size of the trusted code
base.

\paragraph{Security level.}
The security reduction to the discrete logarithm problem is
almost identical to the case of Schnorr signatures~\cite{PS96}.
Notably, the challenge space has about half the size 
(\(\zplus\) versus \(\Z_N\)) while the proof of soundness
computes either \(s_0+s_1\) or \(s_0-s_1\).
This results in a slightly weaker reduction, as should be expected
by moving from \(\J\) to \(\K\) and by weakening verification.
By choosing \(\log_2 N\approx 256\) we obtain a scheme with about
the same security level as state-of-the-art schemes
(eg. EdDSA combined with Ed25519).
This could be made more precise (cf.~\cite{PS00}),
		 but we do not provide this analysis here.

\paragraph{Key and signature sizes.}
Public keys fit into 32 bytes
in both the genus 1 and genus 2 settings. 
This is standard for Montgomery curves;
for Kummer surfaces it requires a new compression technique,
which we present in~\S\ref{sec:comp}. 
In both cases \(\log_2 N<256\),
which means that signatures \((\pm\R\concat s)\) fit in 64 bytes.

\paragraph{Twist security.}
Rational points on \(\K\) correspond to pairs of points 
on either \(\J\) or its quadratic twist.
As opposed to Diffie--Hellman,
in \xsign scalar multiplications with
secret scalars are \emph{only} performed on the public parameter \(\pm\Pc\),
which is chosen as the image of large prime order element of \(\J\).
Therefore \(\J\) is not technically required to have a secure twist,
unlike in the modern Diffie--Hellman setting.
But if \(\K\) is also used for key exchange
(which is the whole point!),
then twist security is crucial.
We therefore strongly recommend twist-secure parameters for \xsign
implementations.

\paragraph{Hash function.} 
The function \(H\) can be any hash function with at least a \(\log_2\sqrt{N}\)-bit security level
and at least \(2\log_2 N\)-bit output.
Throughout this paper we take \(H\)
to be the extendable output function {\tt SHAKE128}~\cite{FIPS202} 
with fixed 512-bit output.
This enables us to implicitly use \(H\) as a function mapping
into either 
\(\Z_N\times \{0,1\}^{256}\) 
(\eg Line~\ref{line:dprime} of Algorithm~\ref{alg:xsign}),
\(\Z_N\) 
(\eg Line~\ref{line:r} of Algorithm~\ref{alg:xsign}),
or
\(\zplus\) 
(\eg Line~\ref{line:h} of Algorithm~\ref{alg:xsign},
by combining it with a conditional negation)
by appropriately reducing (part of) the output modulo \(N\).

\paragraph{Signature compression.} 
Schnorr mentions 
in~\cite{schnorr} 
that signatures \(\left(\R\concat s\right)\) 
may be compressed to \(\left(H(\R\concat\Q\concat M)\concat s\right)\),
taking only the first 128 bits of the hash,
thus reducing signature size from 64 to 48 bytes.
This is possible because we can recompute \(\R\) 
from \(\Pc\), \(\Q\), \(s\), and \(H(\R\concat\Q\concat M)\).
However, on \(\K\) we cannot recover \(\pm\R\) from \(\pm\Pc\),
\(\pm\Q\), \(s\), and 
\(H(\pm\R\concat\pm\Q\concat M)\), 
so Schnorr's compression technique is not an option for us.

\paragraph{Batching.}
Proposals for batch signature verification typically 
rely on the group structure,
verifying random linear combinations of points~\cite{NMVR94,EdDSA}.
Since \(\K\) has no group structure, these batching algorithms
are not possible.

\paragraph{Scalar multiplication for verification.}
Instead of computing the full point \([s]\Pc+[c]\Q\) 
with a two-dimensional multiscalar multiplication operation, 
we have to compute \(\pm[s]\Pc\) and \(\pm[c]\Q\) separately. 
As a result we are unable to use standard tricks for speeding
up two-dimensional scalar multiplications (\eg~\cite{G84}), 
resulting
in increased run-time.
On the other hand, it has the benefit of relying on the already implemented
{\tt Ladder} function, mitigating the need for a separate algorithm,
and is more memory-friendly.
Our implementations show a significant decrease in stack usage,
at the cost of a small loss of speed (see~\S\ref{sec:imp}).

%% file: elliptic.tex
\section{
    Implementing \xsign with elliptic curves
}
\label{sec:elliptic}

Our first concrete instantiation of \xsign
uses the Kummer variety of an elliptic curve,
which is just the \(x\)-line \(\PP^1\).

\subsection{Montgomery curves}\label{subsec:ecmont}

Consider the elliptic curve in Montgomery form
\[
	E_{AB}/\FF_p : By^2 = x(x^2 + Ax + 1)
    \,,
\]
where \(A^2\neq4\) and \(B\neq0\).
The map \(E_{AB} \to \K=\PP^1\)
defined by
\[
	\Pc = (X : Y : Z)
    \longmapsto
	\pm\Pc 
    =
	\begin{cases}
		(X : Z) &\text{if } Z\neq0 \\
		(1 : 0) &\text{if } Z=0
	\end{cases}
\]
gives rise to efficient \(x\)-only arithmetic on \(\PP^1\) (see~\cite{M87}).
We use the {\tt Ladder} specified in~\cite[Alg. 1]{DHH+15}.
Compression uses Bernstein's map
\[
	{\tt Compress} : 
    (X : Z) \in \PP^1(\FF_p) 
    \longmapsto 
    XZ^{p-2} \in \FF_p 
    \,,
\]
while decompression is the near-trivial
\[
    {\tt Decompress} :
    x \in \FF_p 
    \longmapsto 
    (x : 1) \in \PP^1(\FF_p)
    \,.
\]
Note that {\tt Decompress} never returns \(\bot\),
and that 
\({\tt Decompress}({\tt Compress}((X : Z))) = (X:Z)\) whenever \(Z\neq0\)
(however, the points~\((0:1)\) and~\((1:0)\)
should never appear as public keys or signatures).

\subsection{Signature verification}\label{subsec:ecverif}

It remains to define 
the \(\texttt{Check}\) operation for Montgomery curves.
In the final step of verification 
we are given \(\pm\R\), \(\pm\Pc\), and \(\pm\Q\) in \(\PP^1\),
and we need to check whether
\(\pm\R \in \left\{\pm(\Pc+\Q),\pm(\Pc-\Q)\right\}\).
Proposition~\ref{prop:elliptic-verif}
reduces this to checking a quadratic relation 
in the coordinates of \(\pm\R\), \(\pm\Pc\), and \(\pm\Q\).

\begin{proposition}
    \label{prop:elliptic-verif}
    Writing \((X^P:Z^P) = \pm P\) for \(P\) in \(E_{AB}\), etc.:
    If 
    \(P\), \(Q\), and \(R\) are points on \(E_{AB}\),
    then 
    \(
        \pm R \in \big\{ {\pm(P+Q)}, {\pm(P-Q)} \big\}
    \)
    if and only if
    \begin{equation}
        \label{eq:elliptic-verif}
        B_{ZZ}(X^R)^2 - 2B_{XZ}X^RZ^R + B_{XX}(Z^R)^2 = 0
    \end{equation}
    where
    \begin{align}
        \label{eq:B_XX}
        B_{XX} & = \big(X^PX^Q - Z^PZ^Q\big)^2
        \,,
        \\
        \label{eq:B_XZ}
        B_{XZ}
        & = 
        \big(X^PX^Q + Z^PZ^Q\big)
        \big(X^PZ^Q + Z^PX^Q\big) 
        + 
        2AX^PZ^PX^QZ^Q
        \,,
        \\
        \label{eq:B_ZZ}
        B_{ZZ} & = \big(X^PZ^Q - Z^PX^Q\big)^2
        \,.
    \end{align}
\end{proposition}
\begin{proof}
    Let
    \(S = (X^S:Z^S) = \pm(P+Q)\)
    and
    \(D = (X^D:Z^D) = \pm(P-Q)\).
    If we temporarily assume
    \(\pm 0 \not= \pm P \not= \pm Q \not= \pm 0\)
    and put \(x_P = X^P/Z^P\), etc., then
    the group law on \(E_{AB}\) gives us
    \( 
        x_Sx_D  = (x_Px_Q - 1)^2/(x_P - x_Q)^2 
    \)
    and
    \(
        x_S + x_D = 2((x_Px_Q + 1)(x_P + x_Q) + 2Ax_Px_Q)
    \).
    Homogenizing, we obtain
    \begin{align}
        \label{eq:elliptic-verif-lhs}
        \left(
            X^SX^D 
            :
            X^SZ^D+Z^SX^D 
            :
            Z^SZ^D 
        \right)
        &= 
        \left(
            \lambda B_{XX} 
            :
            \lambda 2B_{XZ}
            :
            \lambda B_{ZZ}
        \right)
        \,.
    \end{align}
    One readily verifies that Equation~\eqref{eq:elliptic-verif-lhs}
    still holds even when the temporary assumption does not 
    (that is, when \(\pm P = \pm Q\) or \(\pm P = \pm 0\) or \(\pm Q = \pm 0\)).
    Having degree 2,
    the homogeneous polynomial \(B_{ZZ}X^2 - B_{XZ}XZ + B_{XX}Z^2\)
    cuts out two points in \(\PP^1\)
    (which may coincide);
    by Equation~\eqref{eq:elliptic-verif-lhs},
    they are \(\pm(P+Q)\)
    and \(\pm(P-Q)\),
    so if \((X^R:Z^R)\) satisfies Equation~\eqref{eq:elliptic-verif}
    then it must be one of them.
    \qed
\end{proof}

\vspace{-0.2cm}
\begin{algorithm}[h]
    \setstretch{1.1}
    \caption{Checking the verification relation for \(\PP^1\)}
    \label{alg:mont-check}
    \SetKwProg{Function}{function}{}{}
    \SetKwInOut{Cost}{Cost}
    \SetKwData{LHS}{LHS}
    \Function{{\tt Check}}{
        \KwIn{%
            \(\pm \Pc\), \(\pm \Q\), \(\pm \R=(x:1)\) 
                            in \(\PP^1\)
                            images of points
                            of \(E_{AB}(\FF_p)\)
        }
        \KwOut{%
            {\bf True} if \(\pm\R \in \{\pm(\Pc+\Q), \pm(\Pc-\Q)\}\),
            {\bf False} otherwise
        }
        \Cost{%
            \(
            8\MUL +
            3\SQR + 
            1\MLC + 
            8\ADD +
            4\SUB
            \)
        }
        \(
            (B_{XX}, B_{XZ}, B_{ZZ})
            \gets 
                        {\tt BValues}(\pm\Pc,\pm\Q)
        \)
        \;
        %
        %
        \lIf{\( B_{XX}x^2 - B_{XZ}x + B_{ZZ} = 0 \)}{%
                \Return{{\rm \bf True}}
        }
        \lElse{\Return{{\rm\bf False}}}
    }
    \Function{{\tt BValues}}{
        \KwIn{%
            \(\pm \Pc=(X^P:Z^P)\), \(\pm \Q=(X^Q:Z^Q)\) in \(\K(\FF_p)\)
        }
        \KwOut{%
            \(\left(B_{XX}(\pm P,\pm Q),B_{XZ}(\pm P,\pm Q),B_{ZZ}(\pm
            P, \pm Q)\right)\) in \(\FF_p^3\)
        }
        \Cost{%
            \(
            6\MUL +
            2\SQR + 
            1\MLC + 
            7\ADD +
            3\SUB
            \)
        }
        \ifapp
        \tcp{%
            See Algorithm~\ref{alg:bvalues}
            and Proposition~\ref{prop:elliptic-verif}
        }
        \else
        \tcp{%
            Use Equations~\eqref{eq:B_XX},
            \eqref{eq:B_XZ},
            and~\eqref{eq:B_ZZ}
            in Proposition~\ref{prop:elliptic-verif}
        }
        \fi
    }
\end{algorithm}

\subsection{Using cryptographic parameters}\label{subsec:ecparam}

We use the elliptic curve 
\(E / \FF_p : y^2 = x^3 + 486662x^2 + x\)
where \(p=2^{255}-19\),
which is commonly 
referred to as {\tt Curve25519}~\cite{Ber06}.
Let \(\Pc\in E(\FF_p)\) be such that \(\pm\Pc=(9:1)\).
Then \(\Pc\) has order \(8N\), where 
\[
    N = 2^{252} + 27742317777372353535851937790883648493
\]
is prime.
The \xDBLADD operation requires us to store \((A+2)/4=121666\),
and we implement optimized multiplication by this constant.
In~\cite[\S3]{Ber06} Bernstein
sets and clears some bits of the private key, 
also referred to as ``clamping''.
This is not necessary in \xsign, 
but we do it anyway in {\tt keypair}
for compatibility.

%% file: kummer.tex
\section{
    Implementing \xsign with Kummer surfaces
}
\label{sec:kummer-arithmetic}

A number of cryptographic protocols
that have been successfully implemented with Montgomery curves
have seen substantial practical improvements
when the curves are replaced with \emph{Kummer surfaces}.
From a general point of view,
a Kummer surface is 
the quotient of some genus-2 Jacobian \(\Jac\) by \(\pm 1\);
geometrically it is a
surface in \(\PP^3\) with sixteen point singularities, called \emph{nodes},
which are the images in \(\Kum\) of the 2-torsion points of \(\Jac\)
(since these are precisely the points fixed by \(-1\)).
From a cryptographic point of view,
a Kummer surface is just a 2-dimensional analogue
of the \(x\)-coordinate used in Montgomery curve arithmetic.

The algorithmic and software aspects of efficient Kummer surface arithmetic
have already been covered in great detail elsewhere
(see eg.~\cite{Gaudry}, \cite{BCLS14}, and~\cite{RSSB16}).
Indeed, the Kummer scalar multiplication algorithms and software 
that we use in our signature implementation
are identical to those described in~\cite{RSSB16},
and use the cryptographic parameters proposed by 
Gaudry and Schost~\cite{gaudry-schost}.

This work includes two entirely new Kummer algorithms
that are essential for our signature scheme:
verification relation testing 
(\texttt{Check}, Algorithm~\ref{alg:kummer-check})
and compression/decompression
(\texttt{Compress} and \texttt{Decompress}, 
Algorithms~\ref{alg:compress} and~\ref{alg:decompress}).
Both of these new techniques require a fair amount of technical
development,
which we begin in this section 
by recalling the basic Kummer equation and constants,
and deconstructing the pseudo-doubling operation into a sequence of
surfaces and maps that will play important roles later.
Once the scene has been set,
we will describe our signature verification algorithm in~\S\ref{sec:verif}
and our point compression scheme in~\S\ref{sec:comp}.
The reader primarily interested in the resulting performance
improvements may wish to skip directly to~\S\ref{sec:imp} on first reading.

The {\tt Check}, {\tt Compress}, and {\tt Decompress} 
algorithms
defined below
require the following subroutines:
\begin{itemize}
    \ifapp 
    \item \(\PPmul\) implements a 4-way parallel multiplication.
        It takes a pair of vectors 
        \((x_1,x_2,x_3,x_4)\)
        and
        \((y_1,y_2,y_3,y_4)\)
        in \(\FF_p^4\),
        and returns \((x_1y_1,x_2y_2,x_3y_3,x_4y_4)\).
    \fi
    \ifapp 
    \item
        \(\PPsqr\) implements a 4-way parallel squaring.
        Given a vector 
        \((x_1,x_2,x_3,x_4)\) in \(\FF_p^4\),
        it returns \((x_1^2,x_2^2,x_3^2,x_4^2)\).
    \fi
    \item
        \(\Hadamard\) implements a Hadamard transform.
        Given a vector 
        \((x_1,x_2,x_3,x_4)\) in \(\FF_p^4\),
        it returns 
        \(
            (
            x_1 + x_2 + x_3 + x_4,
            x_1 + x_2 - x_3 - x_4,
            x_1 - x_2 + x_3 - x_4,
            x_1 - x_2 - x_3 + x_4
            )
        \).
    \item
        \(\Dotprod\) computes the sum of a 4-way multiplication.
        Given a pair of vectors 
        \((x_1,x_2,x_3,x_4)\)
        and
        \((y_1,y_2,y_3,y_4)\)
        in \(\FF_p^4\),
        it returns \(x_1y_1+x_2y_2+x_3y_3+x_4y_4\).
\end{itemize}

\subsection{Constants}
\label{sec:constants}

Our Kummer surfaces are defined by four fundamental constants
\(\an\), \(\bn\), \(\cn\), \(\dn\)
and four dual constants \(\ad\), \(\bd\), \(\cd\), and \(\dd\),
which are related by
\begin{align*}
    2\ad^2 & = \an^2 + \bn^2 + \cn^2 + \dn^2 \,,
    \\
    2\bd^2 & = \an^2 + \bn^2 - \cn^2 - \dn^2 \,,
    \\
    2\cd^2 & = \an^2 - \bn^2 + \cn^2 - \dn^2 \,,
    \\
    2\dd^2 & = \an^2 - \bn^2 - \cn^2 + \dn^2 \,.
\end{align*}
We require all of the \(\unsquaredn_i\) and \(\unsquaredd_i\) to be
nonzero.
The fundamental constants determine the dual constants up to sign,
and vice versa.
These relations remain true when we exchange
the \(\unsquaredn_i\) with the \(\unsquaredd_i\);
we call this ``swapping \(x\) with \(\dualof{x}\)'' operation ``dualizing''.
To make the symmetry in what follows clear,
we define
\begin{align*}
    \aan & := \an^2 \,,
    &
    \epsilon_1 & := \bbn\ccn\ddn \,,
    &
    \kappa_1 & := \epsilon_1 + \epsilon_2 + \epsilon_3 + \epsilon_4 \,,
    \\
    \bbn & := \bn^2 \,,
    &
    \epsilon_2 & := \aan\ccn\ddn \,,
    &
    \kappa_2 & := \epsilon_1 + \epsilon_2 - \epsilon_3 - \epsilon_4 \,,
    \\
    \ccn & := \cn^2 \,,
    &
    \epsilon_3 & := \mu_1\mu_2\mu_4 \,,
    &
    \kappa_3 & := \epsilon_1 - \epsilon_2 + \epsilon_3 - \epsilon_4 \,,
    \\
    \ddn & := \dn^2 \,,
    &
    \epsilon_4 & := \mu_1\mu_2\mu_3 \,,
    &
    \kappa_4 & := \epsilon_1 - \epsilon_2 - \epsilon_3 + \epsilon_4 \,,
\end{align*}
along with their respective duals 
\(\dualof{\mu}_i\), 
\(\dualof{\epsilon}_i\),
and \(\dualof{\kappa}_i\).
Note that 
\[
    (\epsilon_1:\epsilon_2:\epsilon_3:\epsilon_4)
    =
    (1/\aan:1/\bbn:1/\ccn:1/\ddn)
\]
and
\(
    \squaredn_i\squaredn_j - \squaredn_k\squaredn_l
    =
    \squaredd_i\squaredd_j - \squaredd_k\squaredd_l
\)
for \(\{i,j,k,l\} = \{1,2,3,4\}\).
There are many clashing notational conventions for theta constants
in the cryptographic Kummer literature;
Table~\ref{tab:constant-dictionary} provides a dictionary for converting
between them.

Our applications use only the squared constants 
\(\squaredn_i\) and \(\squaredd_i\),
so only they need be in \(\FF_p\).
In practice we want them to be as ``small'' as possible,
both to reduce the cost of multiplying by them
and to reduce the cost of storing them.
In fact, it follows from their definition
that it is much easier to find simultaneously small \(\squaredn_i\)
and \(\squaredd_i\) than it is to find simultaneously small
\(\unsquaredn_i\) and \(\unsquaredd_i\)
(or a mixture of the two);
this is ultimately why we prefer the squared surface 
for scalar multiplication.
We note that if the \(\squaredn_i\) are very small,
then the \(\epsilon_i\) and \(\kappa_i\)
are also small,
and the same goes for their duals.
While we will never actually compute with the unsquared constants,
we need them to explain what is happening in the background below.

Finally, 
the Kummer surface equations involve some derived constants
\begin{align*}
    \En 
    & := 
    \frac{
        16 \an\bn\cn\dn \aad\bbd\ccd\ddd
    }{
        (\aan\ddn-\bbn\ccn)(\aan\ccn-\bbn\ddn)(\aan\bbn-\ccn\ddn)
    }
    \,,
\end{align*}
\begin{align*}
    \Fn 
    & := 
    2\frac{\aan\ddn+\bbn\ccn}{\aan\ddn-\bbn\ccn}
    \,,
    &
    \Gn 
    & := 
    2\frac{\aan\ccn+\bbn\ddn}{\aan\ccn-\bbn\ddn}
    \,,
    &
    \Hn 
    & := 
    2\frac{\aan\bbn+\ccn\ddn}{\aan\bbn-\ccn\ddn}
    \,,
\end{align*}
and their duals \(\Ed\), \(\Fd\), \(\Gd\), \(\Hd\).
We observe that 
\( \En^2 = \Fn^2 + \Gn^2 + \Hn^2 + \Fn\Gn\Hn - 4 \)
and
\( \Ed^2 = \Fd^2 + \Gd^2 + \Hd^2 + \Fd\Gd\Hd - 4 \).

\begin{table}
    \centering
    \begin{tabular}{c|c|c}
        Source & Fundamental constants & Dual constants \\
        \hline
        \cite{Gaudry}
        and
        \cite{BCLS14}
        &
        \((a\!:\!b\!:\!c\!:\!d)=(\an\!:\!\bn\!:\!\cn\!:\!\dn)\) 
        &
        \((A\!:\!B\!:\!C\!:\!D)=(\ad\!:\!\bd\!:\!\cd\!:\!\dd)\)
        \\
        \cite{BCHL13}
        &
        \((a\!:\!b\!:\!c\!:\!d)=(\an\!:\!\bn\!:\!\cn\!:\!\dn)\) 
        &
        \((A\!:\!B\!:\!C\!:\!D)=(\aad\!:\!\bbd\!:\!\ccd\!:\!\ddd)\)
        \\
        \cite{RSSB16} 
        &
        \((a\!:\!b\!:\!c\!:\!d)=(\aan\!:\!\bbn\!:\!\ccn\!:\!\ddn)\)
        &
        \((A\!:\!B\!:\!C\!:\!D)=(\aad\!:\!\bbd\!:\!\ccd\!:\!\ddd)\)
        \\
        \cite{cosset}
        &
        \((\alpha\!:\!\beta\!:\!\gamma\!:\!\delta)=(\aan\!:\!\bbn\!:\!\ccn\!:\!\ddn)\)
        &
        \((A\!:\!B\!:\!C\!:\!D)=(\aad\!:\!\bbd\!:\!\ccd\!:\!\ddd)\)
        \\
    \end{tabular}
    \caption{Relations between our theta constants and others in
    selected related work}
    \label{tab:constant-dictionary}
\end{table}

\subsection{Fast Kummer surfaces}
\label{subsec:fastkummer}

We compute all of the pseudoscalar multiplications in \xsign
on the so-called \textbf{squared Kummer surface}
\[
    \KChudnovsky :
    4\En^2\cdot\XCn_1\XCn_2\XCn_3\XCn_4 
    = 
    \left( 
        \begin{array}{c}
            \XCn_1^2 + \XCn_2^2 + \XCn_3^2 + \XCn_4^2
            - \Fn(\XCn_1\XCn_4+\XCn_2\XCn_3)
            \\
            {}
            - \Gn(\XCn_1\XCn_3 + \XCn_2\XCn_4) 
            - \Hn(\XCn_1\XCn_2 + \XCn_3\XCn_4) 
        \end{array}
    \right)^2
    \,,
\]
which was proposed for factorization algorithms
by the Chudnovskys~\cite{Chudnovsky--Chudnovsky},
then later for Diffie--Hellman by Bernstein~\cite{B-ecc06-talk}. 
Since \(\En\) only appears as a square,
\(\KChudnovsky\) is defined over \(\FF_p\).
The zero point on \(\KChudnovsky\) is \(\pm 0 = (\mu_1:\mu_2:\mu_3:\mu_4)\).
In our implementations
we used the \(\xDBLADD\) and Montgomery ladder
exactly as they were presented in~\cite[Algorithms~6-7]{RSSB16}%
\ifapp
\xspace (see also Algorithm~\ref{alg:ladder-kummer}).
\else
.
\fi
The pseudo-doubling \(\xDBL\)
on \(\KChudnovsky\)
is
\[
    \pm P = \big(\XCn_1^P:\XCn_2^P:\XCn_3^P:\XCn_4^P\big)
    \longmapsto
    \big(\XCn_1^{[2]P}:\XCn_2^{[2]P}:\XCn_3^{[2]P}:\XCn_4^{[2]P}\big)
    = \pm[2]P
\]
where
\begin{align}
    \label{eq:xDBL-1}
    \XCn_1^{[2]P} & = \epsilon_1(U_1 + U_2 + U_3 + U_4)^2
    \,,
    &
    U_1 & = \epsilond_1(\XCn_1^P + \XCn_2^P + \XCn_3^P + \XCn_4^P)^2
    \,,
    \\
    \label{eq:xDBL-2}
    \XCn_2^{[2]P} & = \epsilon_2(U_1 + U_2 - U_3 - U_4)^2
    \,,
    &
    U_2 & = \epsilond_2(\XCn_1^P + \XCn_2^P - \XCn_3^P - \XCn_4^P)^2
    \,,
    \\
    \label{eq:xDBL-3}
    \XCn_3^{[2]P} & = \epsilon_3(U_1 - U_2 + U_3 - U_4)^2
    \,,
    &
    U_3 & = \epsilond_3(\XCn_1^P - \XCn_2^P + \XCn_3^P - \XCn_4^P)^2
    \,,
    \\
    \label{eq:xDBL-4}
    \XCn_4^{[2]P} & = \epsilon_4(U_1 - U_2 - U_3 + U_4)^2
    \,,
    &
    U_4 & = \epsilond_4(\XCn_1^P - \XCn_2^P - \XCn_3^P + \XCn_4^P)^2
\end{align}
for \(\pm P\) with all \(\XCn_i^P \not= 0\);
more complicated formul\ae{} exist for other \(\pm P\) 
(cf.~\S\ref{sec:forms}).

\subsection{Deconstructing pseudo-doubling}
\label{sec:hexagon}

Figure~\ref{fig:hexagon}
deconstructs the pseudo-doubling on \(\KChudnovsky\)
from~\S\ref{subsec:fastkummer}
into a cycle of atomic maps between different Kummer surfaces,
which form a sort of hexagon.
\begin{figure}
    \[
        \xymatrix{
            & 
            \KGaudry \ar[r]^{\Squaren}_{(2,2)} 
            & 
            \KChudnovsky \ar[dr]^{\Hadamardn}_{\cong}
            &
            \\
            \KInterdual \ar[ur]^{\Scalen}_{\cong}
            & 
            & 
            & 
            \KInter \ar[dl]^{\Scaled}_{\cong}
            \\
            & 
            \KChudnovskydual \ar[ul]^{\Hadamardd}_{\cong}
            & 
            \KGaudrydual \ar[l]^{\Squared}_{(2,2)}
            &
        }
    \]
    \caption{
        Decomposition of pseudo-doubling on fast Kummer surfaces
        into a cycle of morphisms.
        Here,
        \(\KChudnovsky\) is the ``squared'' surface we mostly compute with;
        \(\KGaudry\) is the related ``canonical'' surface;
        and \(\KInter\) is a new ``intermediate'' surface
        which we use in signature verification.
        (The surfaces \(\KChudnovskydual\), 
        \(\KGaudrydual\),
        and \(\KInterdual\)
        are their duals.)
    }
    \label{fig:hexagon}
\end{figure}
Starting at any one of the Kummers 
and doing a complete cycle of these maps 
carries out pseudo-doubling on that Kummer.
Doing a half-cycle from a given Kummer around to its dual
computes a \((2,2)\)-isogeny splitting pseudo-doubling.

Six different Kummer surfaces may seem like a lot to keep track of---even if
there are really only three, together with their duals.
However, the new surfaces are important,
because they are crucial in deriving our \texttt{Check} routine
(of course, once the algorithm has been written down,
the reader is free to forget about the existence of these other surfaces).

The cycle actually begins 
one step before \(\KChudnovsky\),
with the \textbf{canonical surface}
\[
    \KGaudry
    : 
        2\En\cdot\XGn_1\XGn_2\XGn_3\XGn_4 
    = 
    \begin{array}{c}
        \XGn_1^4 + \XGn_2^4 + \XGn_3^4 + \XGn_4^4
        - \Fn(\XGn_1^2\XGn_4^2 + \XGn_2^2\XGn_3^2) 
        \\
        {}
        - \Gn(\XGn_1^2\XGn_3^2 + \XGn_2^2\XGn_4^2) 
        - \Hn(\XGn_1^2\XGn_2^2 + \XGn_3^2\XGn_4^2) 
        \,.
    \end{array}
\]
This was the model proposed for cryptographic applications by Gaudry
in~\cite{Gaudry};
we call it ``canonical'' because 
it is the model arising from a canonical basis of theta functions of
level \((2,2)\). 

Now we can begin our tour around the hexagon,
moving from \(\KGaudry\) to \(\KChudnovsky\)
via the \textbf{squaring} map
\[
    \Squaren : 
    \big(\XGn_1:\XGn_2:\XGn_3:\XGn_4\big)
    \longmapsto
    \big(\XCn_1:\XCn_2:\XCn_3:\XCn_4\big)
    =
    \big(\XGn_1^2:\XGn_2^2:\XGn_3^2:\XGn_4^3\big)
    \,,
\]
which corresponds to a \((2,2)\)-isogeny of Jacobians.
Moving on from \(\KChudnovsky\),
the \textbf{Hadamard transform} isomorphism
\[
    \Hadamardn 
    : 
    \left(\XCn_1:\XCn_2:\XCn_3:\XCn_4\right)
    \longmapsto
    \left(\XHn_1:\XHn_2:\XHn_3:\XHn_4\right)
    =
    \left(
        \begin{array}{r}
            \XCn_1+\XCn_2+\XCn_3+\XCn_4
            \\
            : \XCn_1+\XCn_2-\XCn_3-\XCn_4
            \\
            : \XCn_1-\XCn_2+\XCn_3-\XCn_4
            \\
            : \XCn_1-\XCn_2-\XCn_3+\XCn_4
        \end{array}
    \right) 
\] 
takes us into a third kind of Kummer,
which we call the \textbf{intermediate surface}:
\[
    \KInter
    : 
    \frac{2\Ed}{\an\bn\cn\dn}\cdot\XHn_1\XHn_2\XHn_3\XHn_4
    =
    \begin{array}{c}
        \frac{\XHn_1^4}{\mu_1^2} + \frac{\XHn_2^4}{\mu_2^2}
        + \frac{\XHn_3^4}{\mu_3^2} + \frac{\XHn_4^4}{\mu_4^2}
        - \Fd\left(
            \frac{\XHn_1^2}{\mu_1}
            \frac{\XHn_4^2}{\mu_4}
            +
            \frac{\XHn_2^2}{\mu_2}
            \frac{\XHn_3^2}{\mu_3}
        \right)
        \\
        {} 
        - \Gd\left(
            \frac{\XHn_1^2}{\mu_1}
            \frac{\XHn_3^2}{\mu_3}
            +
            \frac{\XHn_2^2}{\mu_2}
            \frac{\XHn_4^2}{\mu_4}
        \right)
        - 
        \Hd\left(
            \frac{\XHn_1^2}{\mu_1}
            \frac{\XHn_2^2}{\mu_2}
            +
            \frac{\XHn_3^2}{\mu_3}
            \frac{\XHn_4^2}{\mu_4}
        \right)
        \,.
    \end{array}
\]
We will use \(\KInter\) for signature verification.
Now the \textbf{dual scaling} isomorphism
\[
    \Scaled : 
    \big(\XHn_1:\XHn_2:\XHn_3:\XHn_4\big)
    \longmapsto 
    \big(\XGd_1:\XGd_2:\XGd_3:\XGd_4\big)
    =
    \big(\XHn_1/\ad:\XHn_2/\bd:\XHn_3/\cd:\XHn_4/\dd\big)
\]
takes us into the 
\textbf{dual canonical surface}
\[
    \KGaudrydual
    : 
    2\Ed\cdot\XGd_1\XGd_2\XGd_3\XGd_4 
    =
    \begin{array}{c}
        \XGd_1^4 + \XGd_2^4 + \XGd_3^4 + \XGd_4^4
        - \Fd(\XGd_1^2\XGd_4^2 + \XGd_2^2\XGd_3^2) 
        \\
        {}
        - \Gd(\XGd_1^2\XGd_3^2 + \XGd_2^2\XGd_4^2) 
        - \Hd(\XGd_1^2\XGd_2^2 + \XGd_3^2\XGd_4^2) 
        \,.
    \end{array}
\]
We are now halfway around the hexagon;
the return journey is simply the dual of the outbound trip.
The \textbf{dual squaring}
map
\[
    \Squared : 
    \big(\XGd_1:\XGd_2:\XGd_3:\XGd_4\big)
    \longmapsto
    \big(\XCd_1:\XCd_2:\XCd_3:\XCd_4\big)
    =
    \big(\XGd_1^2:\XGd_2^2:\XGd_3^2:\XGd_4^3\big)
    \,,
\]
another \((2,2)\)-isogeny,
carries us into the \textbf{dual squared surface}
\[
    \KChudnovskydual
    : 
    4\Ed^2\cdot\XCd_1\XCd_2\XCd_3\XCd_4 
    = 
    \left( 
        \begin{array}{c}
            \XCd_1^2 + \XCd_2^2 + \XCd_3^2 + \XCd_4^2
            - \Fd(\XCd_1\XCd_4+\XCd_2\XCd_3)
            \\
            {}
            - \Gd(\XCd_1\XCd_3 + \XCd_2\XCd_4) 
            - \Hd(\XCd_1\XCd_2 + \XCd_3\XCd_4) 
        \end{array}
    \right)^2
    \,,
\]
before the \textbf{dual Hadamard transform}
\[
    \Hadamardd 
    : 
    \big(\XCd_1:\XCd_2:\XCd_3:\XCd_4\big)
    \longmapsto
    \big(\XHd_1:\XHd_2:\XHd_3:\XHd_4\big)
    =
    \left(
        \begin{array}{r}
            \XCd_1+\XCd_2+\XCd_3+\XCd_4
            \\
            : \XCd_1+\XCd_2-\XCd_3-\XCd_4
            \\
            : \XCd_1-\XCd_2+\XCd_3-\XCd_4
            \\
            : \XCd_1-\XCd_2-\XCd_3+\XCd_4
        \end{array}
    \right) 
\]
takes us into the \textbf{dual intermediate surface}
\[
    \KInterdual
    : 
    \frac{2\En}{\an\bn\cn\dn}\cdot\XHd_1\XHd_2\XHd_3\XHd_4
    =
    \begin{array}{c}
        \frac{\XHd_1^4}{\mu_1^2} + \frac{\XHd_2^4}{\mu_2^2}
        + \frac{\XHd_3^4}{\mu_3^2} + \frac{\XHd_4^4}{\mu_4^2}
        - \Fd\left(
            \frac{\XHd_1^2}{\mu_1}
            \frac{\XHd_4^2}{\mu_4}
            - 
            \frac{\XHd_2^2}{\mu_2}
            \frac{\XHd_3^2}{\mu_3}
        \right)
        \\
        {} 
        - \Gd\left(
            \frac{\XHd_1^2}{\mu_1}
            \frac{\XHd_3^2}{\mu_3}
            - 
            \frac{\XHd_2^2}{\mu_2}
            \frac{\XHd_4^2}{\mu_4}
        \right)
        - 
        \Hd\left(
            \frac{\XHd_1^2}{\mu_1}
            \frac{\XHd_2^2}{\mu_2}
            - 
            \frac{\XHd_3^2}{\mu_3}
            \frac{\XHd_4^2}{\mu_4}
        \right)
        \,.
    \end{array}
\]
A final \textbf{scaling} isomorphism
\[
    \Scalen : 
    \big(\XHd_1:\XHd_2:\XHd_3:\XHd_4\big)
    \longmapsto 
    \big(\XGn_1:\XGn_2:\XGn_3:\XGn_4\big)
    =
    \big(
        \XHd_1/\an:\XHd_2/\bn:\XHd_3/\cn:\XHd_4/\dn
    \big)
\]
takes us from \(\KInterdual\) back to \(\KGaudry\), where we started.

The canonical surfaces
\(\KGaudry\) resp. \(\KGaudrydual\) 
are only defined over \(\FF_p(\an\bn\cn\dn)\) resp. \(\FF_p(\ad\bd\cd\dd)\),
while the scaling isomorphisms \(\Scaled\) resp. \(\Scalen\)
are defined over \(\FF_p(\ad,\bd,\cd,\dd)\) 
resp. \(\FF_p(\an,\bn,\cn,\dn)\).
Everything else is defined over \(\FF_p\).

We confirm that one cycle around the hexagon,
starting and ending on \(\KChudnovsky\), 
computes the pseudo-doubling 
of Equations~\eqref{eq:xDBL-1}, 
\eqref{eq:xDBL-2},
\eqref{eq:xDBL-3},
and~\eqref{eq:xDBL-4}.
Similarly,
one cycle around the hexagon starting and ending on \(\KGaudry\)
computes Gaudry's pseudo-doubling from~\cite[\S3.2]{Gaudry}.

%% file: verification.tex
\section{Signature verification on Kummer surfaces}
\label{sec:verif}

To verify signatures in the Kummer surface implementation,
we need to supply a \texttt{Check} algorithm
which,
given \(\pm P\), \(\pm Q\), and \(\pm R\) on \(\KChudnovsky\),
decides whether
\( \pm R \in \{ \pm (P+Q), \pm(P-Q)\} \).
For the elliptic version of \xsign described in~\S\ref{sec:elliptic},
we saw that this came down to 
checking that \(\pm R\) satisfied one quadratic relation
whose three coefficients were biquadratic forms
in \(\pm P\) and \(\pm Q\).
The same principle extends to Kummer surfaces,
where the pseudo-group law 
is similarly defined by biquadratic forms;
but since Kummer surfaces are defined in terms of four coordinates
(as opposed to the two coordinates of the \(x\)-line),
this time there are
six simple quadratic relations to verify,
with a total of ten coefficient forms.

\subsection{Biquadratic forms and pseudo-addition}
\label{sec:forms}

Let \(\Kum\) be a Kummer surface.
If \(\pm P\) is a point on \(\Kum\),
then we write \((Z_1^P:Z_2^P:Z_3^P:Z_4^P)\)
for its projective coordinates.
The classical theory of abelian varieties
tells us that there exist
biquadratic forms \(B_{ij}\) 
for \(1 \le i, j \le 4\)
such that for all
\(\pm P\)
and
\(\pm Q\),
if
\(\pm S = \pm (P + Q)\)
and
\(\pm D = \pm (P - Q)\)
then
\begin{equation}
    \label{eq:kummer-biquad-def}
    \left(
        Z_i^SZ_j^D + Z_j^SZ_i^D 
    \right)_{i,j=1}^4
    =
    \lambda
    \left(
        B_{ij}(Z_1^P,Z_2^P,Z_3^P,Z_4^P,Z_1^Q,Z_2^Q,Z_3^Q,Z_4^Q)
    \right)_{i,j=1}^4
\end{equation}
where \(\lambda\in\field^\times\) is some common projective factor 
depending only on the affine representatives chosen for 
\(\pm P\), \(\pm Q\), \(\pm(P+Q)\) and \(\pm(P-Q)\).
These biquadratic forms are the foundation of
pseudo-addition and doubling laws on \(\Kum\):
if the ``difference'' \(\pm D\) is known, 
then we can use the \(B_{ij}\) to compute \(\pm S\).

\begin{proposition}
    \label{prop:kummer-verif}
    Let \(\{B_{ij}:1\le i,j\le 4\}\)
    be a set of biquadratic forms
    on \(\Kum\times\Kum\)
    satisfying Equation~\eqref{eq:kummer-biquad-def}
    for all \(\pm P\), \(\pm Q\), \(\pm(P+Q)\), and \(\pm(P-Q)\).
    Then
    \[
        \pm R = (Z_1^R:Z_2^R:Z_3^R:Z_4^R) 
        \in
        \left\{
            {\pm(P+Q)}, {\pm(P-Q)}
        \right\}
    \]
    if and only if 
    (writing \(B_{ij}\) for \(B_{ij}(Z_1^P,\ldots,Z_4^Q)\))
    we have
    \begin{equation}
        \label{eq:kummer-verif-lhs}
        B_{jj}\cdot(Z_i^R)^2 
        - 
        2B_{ij}\cdot Z_i^RZ_j^R 
        + 
        B_{ii}\cdot(Z_j^R)^2 
        = 
        0
        \quad
        \text{for all } 1 \le i < j \le 4
        \,.
    \end{equation}
\end{proposition}
\begin{proof}
    Looking at Equation~\eqref{eq:kummer-biquad-def},
    we see that
    the system of six quadratics from Equation~\eqref{eq:kummer-verif-lhs}
    cuts out a zero-dimensional degree-2 subscheme of \(\Kum\):
    that is, the pair of points \(\{\pm(P+Q),\pm(P-Q)\}\) (which may coincide).
    Hence, if \((Z_1^R:Z_2^R:Z_3^R:Z_4^R) = \pm R\)
    satisfies all of the equations,
    then it must be one of them.
    \qed
\end{proof}

\subsection{Deriving efficiently computable forms}

Proposition~\ref{prop:kummer-verif}
is the exact analogue of Proposition~\ref{prop:elliptic-verif}
for Kummer surfaces.
All that we need 
to turn it into a \texttt{Check} algorithm for \xsign
is an explicit and efficiently computable
representation of the \(B_{ij}\).
These forms depend
on the projective model of the Kummer surface;
so we write \(\BCan_{ij}\), \(\BSqr_{ij}\), and \(\BInt_{ij}\)
for the forms on the canonical, squared, and intermediate surfaces.

On the canonical surface,
the forms \(\BCan_{ij}\) are classical 
(see e.g.~\cite[\S2.2]{Baily62}).
The on-diagonal forms \(\BCan_{ii}\) are 
\begin{align}
    \label{eq:BCan-ii-1}
    \BCan_{11} & = \frac{1}{4}\Big(\frac{V_1}{\squaredd_1} + \frac{V_2}{\squaredd_2} + \frac{V_3}{\squaredd_3} + \frac{V_4}{\squaredd_4}\Big)
    \,,
    &
    \BCan_{22} & = \frac{1}{4}\Big(\frac{V_1}{\squaredd_1} + \frac{V_2}{\squaredd_2} - \frac{V_3}{\squaredd_3} - \frac{V_4}{\squaredd_4}\Big)
    \,,
    \\
    \label{eq:BCan-ii-2}
    \BCan_{33} & = \frac{1}{4}\Big(\frac{V_1}{\squaredd_1} - \frac{V_2}{\squaredd_2} + \frac{V_3}{\squaredd_3} - \frac{V_4}{\squaredd_4}\Big)
    \,,
    &
    \BCan_{44} & = \frac{1}{4}\Big(\frac{V_1}{\squaredd_1} - \frac{V_2}{\squaredd_2} - \frac{V_3}{\squaredd_3} + \frac{V_4}{\squaredd_4}\Big)
    \,,
\end{align}
where
\begin{align*}
    V_1 
    & = 
    \big((\XGn_1^P)^2 + (\XGn_2^P)^2 + (\XGn_3^P)^2 + (\XGn_4^P)^2\big)
    \big((\XGn_1^Q)^2 + (\XGn_2^Q)^2 + (\XGn_3^Q)^2 + (\XGn_4^Q)^2\big)
    \,,
    \\
    V_2 
    & = 
    \big((\XGn_1^P)^2 + (\XGn_2^P)^2 - (\XGn_3^P)^2 - (\XGn_4^P)^2\big)
    \big((\XGn_1^Q)^2 + (\XGn_2^Q)^2 - (\XGn_3^Q)^2 - (\XGn_4^Q)^2\big)
    \,,
    \\
    V_3 
    & = 
    \big((\XGn_1^P)^2 - (\XGn_2^P)^2 + (\XGn_3^P)^2 - (\XGn_4^P)^2\big)
    \big((\XGn_1^Q)^2 - (\XGn_2^Q)^2 + (\XGn_3^Q)^2 - (\XGn_4^Q)^2\big)
    \,,
    \\
    V_4 
    & = 
    \big((\XGn_1^P)^2 - (\XGn_2^P)^2 - (\XGn_3^P)^2 + (\XGn_4^P)^2\big)
    \big((\XGn_1^Q)^2 - (\XGn_2^Q)^2 - (\XGn_3^Q)^2 + (\XGn_4^Q)^2\big)
    \,,
\end{align*}
while the off-diagonal forms \(B_{ij}\)
with \(i \not= j\) are
\begin{align}
    \label{eq:BCan-ij}
    \BCan_{ij}
    & =
    \frac{
        2
    }{
        \squaredd_i\squaredd_j-\squaredd_k\squaredd_l
    }
    \left(
        \begin{array}{l}
            \unsquaredn_i\unsquaredn_j
            \big(
                \XGn_i^P\XGn_j^P\XGn_i^Q\XGn_j^Q 
                + 
                \XGn_k^P\XGn_l^P\XGn_k^Q\XGn_l^Q
            \big) 
            \\
            {}
            - 
            \unsquaredn_k\unsquaredn_l
            \big(
                \XGn_i^P\XGn_j^P\XGn_k^Q\XGn_l^Q 
                + 
                \XGn_k^P\XGn_l^P\XGn_i^Q\XGn_j^Q
            \big)
        \end{array}
    \right)
\end{align}
where \(\{i,j,k,l\} = \{1,2,3,4\}\).

All of these forms can be efficiently evaluated.
The off-diagonal \(\BCan_{ij}\) have a particularly compact shape,
while the symmetry of the on-diagonal \(\BCan_{ii}\) makes them
particularly easy to compute simultaneously:
indeed, that is exactly what we do in Gaudry's fast pseudo-addition
algorithm for \(\KGaudry\)~\cite[\S3.2]{Gaudry}.

Ideally, we would like to evaluate the \(\BSqr_{ij}\) on \(\KChudnovsky\),
since that is where our inputs \(\pm P\), \(\pm Q\), and \(\pm R\) 
live.
We can compute the \(\BSqr_{ij}\) 
by dualizing the \(\BCan_{ij}\),
then pulling the \(\BCandual_{ij}\) on \(\KGaudrydual\)
back to \(\KChudnovsky\) 
via \(\Scaled\circ\Hadamardn\).
But while the resulting on-diagonal \(\BSqr_{ii}\) 
maintain the symmetry and efficiency of the \(\BCan_{ii}\),\footnote{%
    As they should, since they are the basis of the
    efficient pseudo-addition on \(\KChudnovsky\)!
}
the off-diagonal \(\BSqr_{ij}\)
turn out to be much less pleasant,
with less apparent exploitable symmetry.
For our applications,
this means that evaluating \(\BSqr_{ij}\) for \(i\not=j\)
implies taking a significant hit in terms of stack
and code size,
not to mention time.

%

We could avoid this difficulty
by mapping the inputs of \texttt{Check}
from \(\KChudnovsky\) into \(\KGaudrydual\),
and then evaluating the \(\BCandual_{ij}\).
But this would involve using---and, therefore, storing---%
the four large unsquared \(\unsquaredd_i\),
which is an important drawback.

Why do the nice \(\BCandual_{ij}\)
become so ugly when pulled back to \(\KChudnovsky\)?
The map \(\Scaled:\KInter\to\KGaudrydual\) 
has no impact on the shape or number of monomials,
so most of the ugliness is due to the Hadamard transform 
\(\Hadamardn:\KChudnovsky\to\KInter\).
In particular, 
if we only pull back the \(\BCandual_{ij}\) as far as \(\KInter\),
then the resulting \(\BInt_{ij}\) retain the nice form of the
\(\BCan_{ij}\) but do not involve the \(\unsquaredd_i\). 
This fact prompts our solution:
we map \(\pm P\), \(\pm Q\), and \(\pm R\)
through \(\Hadamardn\) onto \(\KInter\),
and verify using the forms \(\BInt_{ij}\).

\begin{theorem}
    \label{th:KInt-forms}
    Up to a common projective factor,
    the on-diagonal biquadratic forms on the intermediate surface \(\KInter\)
    are
    \begin{align}
        \label{eq:BInt_11}
        \BInt_{11}
        & =
        \squaredd_1\left(
            \kappa_1 F_1 + \kappa_2 F_2 + \kappa_3 F_3 + \kappa_4 F_4
        \right)
        \,,
        \\
        \label{eq:BInt_22}
        \BInt_{22}
        & =
        \squaredd_2\left(
            \kappa_2 F_1 + \kappa_1 F_2 + \kappa_4 F_3 + \kappa_3 F_4
        \right)
        \,,
        \\
        \label{eq:BInt_33}
        \BInt_{33}
        & =
        \squaredd_3\left(
            \kappa_3 F_1 + \kappa_4 F_2 + \kappa_1 F_3 + \kappa_2 F_4
        \right)
        \,,
        \\
        \label{eq:BInt_44}
        \BInt_{44}
        & =
        \squaredd_4\left(
            \kappa_4 F_1 + \kappa_3 F_2 + \kappa_2 F_3 + \kappa_1 F_4
        \right)
        \,,
    \end{align}
    where
    \begin{align*}
        F_1 & = P_1Q_1 + P_2Q_2 + P_3Q_3 + P_4Q_4
        \,,
        &
        F_2 & = P_1Q_2 + P_2Q_1 + P_3Q_4 + P_4Q_3
        \,,
        \\
        F_3 & = P_1Q_3 + P_3Q_1 + P_2Q_4 + P_4Q_2
        \,,
        & 
        F_4 & = P_1Q_4 + P_4Q_1 + P_2Q_3 + P_3Q_2
        \,,
    \end{align*}
    where \( P_i = \dualof{\epsilon}_i (\XHn_i^P)^2 \)
    and \( Q_i  = \dualof{\epsilon}_i (\XHn_i^Q)^2 \)
    for \( 1\le i \le 4\).
    Up to the same common projective factor,
    the off-diagonal forms are 
    \begin{align}
        \label{eq:BInt_ij}
        \BInt_{ij}
        =
        C
        \cdot
        C_{ij}
        \cdot
        \left(
                \squaredd_k\squaredd_l
                \big(\XHn_{ij}^P - \XHn_{kl}^P\big)
                \big(\XHn_{ij}^Q - \XHn_{kl}^Q\big)
                +
                \big(
                    \squaredd_i\squaredd_j 
                    -
                    \squaredd_k\squaredd_l
                \big)
                \XHn_{kl}^P\XHn_{kl}^Q
        \right)
    \end{align}
    for \(\{i,j,k,l\} = \{1,2,3,4\}\)
    where
    \(
        C_{ij} := 
        \squaredd_i\squaredd_j
        (\squaredd_i\squaredd_k - \squaredd_j\squaredd_l)
        (\squaredd_i\squaredd_l - \squaredd_j\squaredd_k)
    \),
    \( \XHn_{ij}^P := \XHn_i^P\XHn_j^P \),
    \( \XHn_{ij}^Q := \XHn_i^Q\XHn_j^Q \),
    and
    \begin{align*}
        C & :=
        \frac{
            8(\squaredn_1\squaredn_2\squaredn_3\squaredn_4)
            (\squaredd_1\squaredd_2\squaredd_3\squaredd_4)
        }{
            (\squaredd_1\squaredd_2-\squaredd_3\squaredd_4)
            (\squaredd_1\squaredd_3-\squaredd_2\squaredd_4)
            (\squaredd_1\squaredd_4-\squaredd_2\squaredd_3)
        }
        \,.
    \end{align*}
\end{theorem}
\begin{proof}
    By definition, 
    \(
        \XGd_i^S\XGd_j^D + \XGd_j^S\XGd_i^D
        =
        \BCandual_{ij}(\XGd_1^P,\ldots,\XGd_4^Q)
    \).
    Pulling back via \(\Scaled\)
    using \(\XGd_i = \XHn_i/\unsquaredd_i\)
    yields
    \begin{align*}
        \BInt_{ij}(\XHn_1^P,\ldots,\XHn_4^Q)
        =
        \XHn_i^S\XHn_j^D + \XHn_j^S\XHn_i^D
        & =
        \unsquaredd_i\unsquaredd_j\big(\XGd_i^S\XGd_j^D + \XGd_j^S\XGd_i^D\big)
        \\
        &
        =
        \unsquaredd_i\unsquaredd_j
        \cdot
        \BCandual_{ij}(\XGd_1^P,\ldots,\XGd_4^Q)
        \\
        &
        =
        \unsquaredd_i\unsquaredd_j
        \cdot
        \BCandual_{ij}(\XHn_1^P/\unsquaredd_1,\ldots,\XHn_4^Q/\unsquaredd_4)
        \,.
    \end{align*}
    Dualizing the \(\BCan_{ij}\)
    from Equations~\eqref{eq:BCan-ii-1}, \eqref{eq:BCan-ii-2},
    and~\eqref{eq:BCan-ij},
    we find
    \begin{align*}
        \BInt_{11} 
        & =
        \squaredd_1 /
        \big(
            4\squaredn_1\squaredn_2\squaredn_3\squaredn_4
            (\squaredd_1\squaredd_2\squaredd_3\squaredd_4)^2
        \big)
        \cdot
        \big( \kappa_1 F_1 + \kappa_2 F_2 + \kappa_3 F_3 + \kappa_4 F_4 \big)
        \,,
        \\
        \BInt_{22} 
        & = 
        \squaredd_2 /
        \big(
            4\squaredn_1\squaredn_2\squaredn_3\squaredn_4
           (\squaredd_1\squaredd_2\squaredd_3\squaredd_4)^2
        \big)
        \cdot
        \big( \kappa_2 F_1 + \kappa_1 F_2 + \kappa_4 F_3 + \kappa_3 F_4 \big)
        \,,
        \\
        \BInt_{33} 
        & =
        \squaredd_3 /
        \big(
            4\squaredn_1\squaredn_2\squaredn_3\squaredn_4
            (\squaredd_1\squaredd_2\squaredd_3\squaredd_4)^2
        \big)
        \cdot
        \big( \kappa_3 F_1 + \kappa_4 F_2 + \kappa_1 F_3 + \kappa_2 F_4 \big)
        \,,
        \\
        \BInt_{44} 
        & = 
        \squaredd_4 /
        \big(
            4\squaredn_1\squaredn_2\squaredn_3\squaredn_4
            (\squaredd_1\squaredd_2\squaredd_3\squaredd_4)^2
        \big)
        \cdot
        \big( \kappa_4 F_1 + \kappa_3 F_2 + \kappa_2 F_3 + \kappa_1 F_4 \big)
        \,,
    \end{align*}
            while the off-diagonal forms \(B_{ij}\) with \(i \not= j\) are
    \begin{align*}
        \BInt_{ij}
        & =
        \frac{
            2
        }{
            \squaredd_k\squaredd_l
            (\squaredd_i\squaredd_j-\squaredd_k\squaredd_l)
        }
        \left(
            \begin{array}{l}
                \squaredd_k\squaredd_l
                \big( Y_{ij}^P - Y_{kl}^P \big)
                \big( Y_{ij}^Q - Y_{kl}^Q \big)
                \\
                {}
                + 
                (\squaredd_i\squaredd_j - \squaredd_k\squaredd_l)
                Y_{kl}^PY_{kl}^Q
            \end{array}
        \right)
    \end{align*}
    for \(\{i,j,k,l\} = \{1,2,3,4\}\).
    Multiplying all of these forms by a common projective factor
    of
    \(
        4(\squaredn_1\squaredn_2\squaredn_3\squaredn_4)
        (\squaredd_1\squaredd_2\squaredd_3\squaredd_4)^2
    \)
    eliminates the denominators in the coefficients,
    and yields the forms of the theorem.
    \qed
\end{proof}

\subsection{Signature verification}

We are now finally ready to implement the \(\texttt{Check}\)
algorithm for \(\KChudnovsky\).
Algorithm~\ref{alg:kummer-check} does this 
by applying \(\Hadamardn\) to its inputs,
then using
the biquadratic forms of Theorem~\ref{th:KInt-forms}.
Its correctness is implied by Proposition~\ref{prop:kummer-verif}.

\begin{algorithm}[ht]
    \setstretch{1.1}
    \caption{Checking the verification relation for points on \(\KChudnovsky\)}
    \label{alg:kummer-check}
    \SetKwProg{Function}{function}{}{}
    \SetKwInOut{Cost}{Cost}
    \SetKwData{V}{V}
    \SetKwData{Y}{Y}
    \SetKwData{P}{P}
    \SetKwData{Q}{Q}
    \SetKwData{F}{F}
    \SetKwData{B}{B}
    \SetKwData{RHS}{RHS}
    \SetKwData{LHS}{LHS}
    \Function{{\tt Check}}{
        \KwIn{%
            \(\pm P\), \(\pm Q\), \(\pm R\) in \(\KChudnovsky(\FF_p)\)
        }
        \KwOut{%
            {\bf True} if \(\pm R \in \{\pm(P+Q), \pm(P-Q)\}\),
            {\bf False} otherwise
        }
        \Cost{%
            \(
            76\MUL + 
            8\SQR + 
            88\MLC + 
            42\ADD + 
            42\SUB
            \)
        }
        \(
            (\Y^P,\Y^Q)
            \gets 
            (\Hadamard(\pm P),\Hadamard(\pm Q))
        \)
        \;
        \(
            (\B_{11},\B_{22},\B_{33},\B_{44})
            \gets
            \texttt{BiiValues}(\Y^P,\Y^Q)
        \)
        \;
        \(
            \Y^R
            \gets 
            \Hadamard(\pm R)
        \)
        \;
        \For{%
            \((i,j)\) 
            \textbf{in}
            \(\{(1,2),(1,3),(1,4),(2,3),(2,4),(3,4)\}\)
        }{
            \(\LHS \gets \B_{ii}\cdot(\Y_j^R)^2  + \B_{jj}\cdot(\Y_i^R)^2 \)
            \;
            \(\B_{ij} \gets \texttt{BijValue}(\Y^P, \Y^Q, (i,j))\)
            \;
            \(\RHS \gets 2\B_{ij}\cdot\Y_i^R\cdot\Y_j^R\)
            \;
            \If{\( \LHS \not= \RHS \)}{%
                \Return{{\bf False}}
            }
        }
        \Return{{\bf True}}
    }
    \Function{{\tt BiiValues}}{
        \KwIn{%
            \(\pm P\), \(\pm Q\) in \(\KInter(\FF_p)\)
        }
        \KwOut{%
            \(
                ( \BInt_{ii}(\pm P, \pm Q) )_{i=1}^4
            \) in \(\FF_p^4\)
        }
        \Cost{%
            \(
            16\MUL + 
            8\SQR + 
            28\MLC + 
            24\ADD
            \)
        }
        \tcp{%
            \ifapp
            See Algorithm~\ref{alg:biivalues} and Theorem~\ref{th:KInt-forms}
            \else
            Use Equations~\eqref{eq:BInt_11}, 
            \eqref{eq:BInt_22},
            \eqref{eq:BInt_33},
            and~\eqref{eq:BInt_44}
            in Theorem~\ref{th:KInt-forms}
            \fi
        }
    }
    \Function{{\tt BijValue}}{
        \KwIn{
            \(\pm P\), \(\pm Q\) in \(\KInter(\FF_p)\)
            and
            \((i,j)\) with \(1\le i,j\le 4\) and \(i \not= j\)
        }
        \KwOut{
            \(\BInt_{ij}(\pm P, \pm Q)\) in \(\FF_p\)
        }
        \Cost{%
            \(
            10\MUL + 
            10\MLC + 
            1\ADD + 
            5\SUB
            \)
        }
        \tcp{%
            \ifapp
            See Algorithm~\ref{alg:bijvalue} and Theorem~\ref{th:KInt-forms}
            \else
            Use Equation~\eqref{eq:BInt_ij} in Theorem~\ref{th:KInt-forms}
            \fi
        }
    }
\end{algorithm}

\subsection{Using cryptographic parameters}

Gaudry and Schost
take \(p = 2^{127}-1\)
and 
\( 
    (\aan:\bbn:\ccn:\ddn)
    =
    (-11:22:19:3)
\) in~\cite{gaudry-schost}.
We also need the constants
\(
    (\aad : \bbd : \ccd : \ddd) = 
    (-33 : 11 : 17 : 49)
\),
\(
    (\kappa_1:\kappa_2:\kappa_3:\kappa_4)
    =
    (-4697:5951:5753:-1991)
\),
and
\(
    (\epsilond_1:\epsilond_2:\epsilond_3:\epsilond_4)
    =
    (-833:2499:1617:561)
\).\footnote{%
    Following the definitions of \S\ref{sec:constants}, 
    the \(\squaredd_i\) are scaled by \(-2\),
    the \(\epsilond_i\) by \(1/11\), and \(C\) by \(2/11^2\).
    These changes influence the \(\BInt_{ij}\), but only up to
    the same projective factor.
}
In practice, where these constants are ``negative'',
we reverse their sign and amend the formul\ae{} above accordingly.
All of these constants are small, 
and fit into one or two bytes each 
(and the \(\epsilond_i\) are already stored for use in \texttt{Ladder}).
We store one large constant 
\[C = {\tt 0x40F50EEFA320A2DD46F7E3D8CDDDA843},\]
and recompute the \(C_{ij}\) on the fly.


%% file: compression.tex
\section{
    Kummer point compression
}
\label{sec:comp}

Our public keys are points on \(\KChudnovsky\),
and each signature includes one point on \(\KChudnovsky\).
Minimizing the space required by Kummer points 
is therefore essential.

A projective Kummer point is composed of four field elements;
normalizing by dividing through by a nonzero coordinate reduces us to
three field elements
(this can also be achieved 
using Bernstein's ``wrapping'' technique~\cite{B-ecc06-talk},
as in~\cite{BCLS14} and~\cite{RSSB16}).
But we are talking about Kummer \emph{surfaces}---two-dimensional
objects---so we might hope to compress to two field elements, plus a few
bits to enable us to correctly recover the whole Kummer point.
This is analogous to elliptic curve point compression, 
where we compress projective points \((X:Y:Z)\)
by normalizing to \((x,y) = (X/Z,Y/Z)\),
then storing \((x,\sigma)\), 
where \(\sigma\) is a bit indicating the ``sign'' of \(y\).
Decompressing the datum \((x,\sigma)\) 
to \((X:Y:Z) = (x:y:1)\) 
then requires solving a simple quadratic
to recover the correct \(y\)-coordinate.

For some reason,
no such Kummer point compression method has explicitly appeared in
the literature.
Bernstein remarked in 2006 that 
if we compress a Kummer point to two coordinates, 
then decompression appears to require solving a complicated quartic
equation~\cite{B-ecc06-talk}.
This would be much more expensive than computing the single square root
required for elliptic decompression;
this has perhaps discouraged implementers from 
attempting to compress Kummer points.

But while it may not always be obvious from their defining equations,
the classical theory tells us that
every Kummer is in fact a double cover of~\(\PP^2\),
just as elliptic curves are double covers of \(\PP^1\).
We use this principle below to show that we can always compress 
any Kummer point to two field elements plus two auxiliary bits,
and then decompress by solving a quadratic.
In our applications,
this gives us a convenient packaging of Kummer points in exactly 256 bits.

\subsection{The general principle}

First, we sketch a general method for Kummer point compression
that works for any Kummer 
presented as a singular quartic surface in \(\PP^3\).

Recall that if \(N\) is any point in \(\PP^3\),
then projection away from \(N\) defines a map \(\pi_N: \PP^3\to\PP^2\)
sending points in \(\PP^3\) on the same line through \(N\)
to the same point in \(\PP^2\).
(The map \(\pi_N\) is only a rational map, and not a morphism;
the image of \(N\) itself is not well-defined.)
Now, let \(N\) be a node of a Kummer surface \(\Kum\):
that is, \(N\) is one of the 16 singular points of \(\Kum\).
The restriction of \(\pi_N\) to \(\Kum\)
forms a double cover of \(\PP^2\).
By definition, \(\pi_N\) maps the points on \(\Kum\)
that lie on the same line through \(N\)
to the same point of \(\PP^2\).
Now \(\Kum\) has degree 4,
so each line in~\(\PP^3\) intersects~\(\Kum\) in four points;
but since \(N\) is a double point of~\(\Kum\), 
every line through \(N\) intersects \(\Kum\) at \(N\) \emph{twice},
and then in two other points.
These two remaining points may be ``compressed'' 
to their common image in \(\PP^2\) under \(\pi_N\),
plus a single bit to distinguish the appropriate preimage.

To make this more concrete, 
let \(L_1\), \(L_2\), and \(L_3\)
be linearly independent linear forms on \(\PP^3\) vanishing on~\(N\);
then \(N\) is the intersection of the three planes in \(\PP^3\)
cut out by the \(L_i\).
We can now realise the projection 
\(\pi_N \colon\Kum\to\PP^2\) 
as
\[
    \pi_N
    \colon
    (P_1:\cdots:P_4) 
    \longmapsto 
    \big(
        L_1(P_1,\ldots,P_4)
        :
        L_2(P_1,\ldots,P_4)
        :
        L_3(P_1,\ldots,P_4)
    \big)
    \,.
\]
Replacing \((L_1,L_2,L_3)\)
with another basis of \(\subgrp{L_1,L_2,L_3}\) 
yields another projection,
which corresponds to composing \(\pi_N\) 
with a linear automorphism of \(\PP^2\).

If \(L_1\), \(L_2\), and \(L_3\) are chosen as above to vanish on \(N\),
and \(L_4\) is any linear form
not in \(\subgrp{L_1,L_2,L_3}\),
then the fact that \(\pi_N\) is a double cover of the \((L_1,L_2,L_3)\)-plane
implies that the defining equation of \(\Kum\)
can be rewritten in the form
\[
    \Kum: K_2(L_1,L_2,L_3)L_4^2
        - 2K_3(L_1,L_2,L_3)L_4
        + K_4(L_1,L_2,L_3)
         = 0
\]
where each \(K_i\) is a homogeneous polynomial 
of degree \(i\) in \(L_1\), \(L_2\), and \(L_3\).
This form, quadratic in \(L_4\),
allows us to replace the \(L_4\)-coordinate
with a single bit indicating the ``sign'' in the corresponding 
root of this quadratic;
the remaining three coordinates can be normalized to an affine plane
point.
The net result is a compression to two field elements,
plus one bit indicating the normalization,
plus another bit to indicate the correct value of \(L_4\).

\begin{remark}
    Stahlke gives a compression algorithm in~\cite{Stahlke04}
    for points on genus-2 Jacobians in the usual Mumford representation.
    The first step can be seen as 
    a projection to the most general model of the Kummer 
    (as in~\cite[Chapter 3]{casselsflynn}),
    and then the second is 
    an implicit implementation of the principle above.
\end{remark}

\subsection{From squared Kummers to tetragonal Kummers}

We want to define an efficient point compression scheme for \(\KChudnovsky\).
The general principle above makes this possible,
but it leaves open the choice of node \(N\)
and the choice of forms \(L_i\).
These choices determine the complexity of the resulting~\(K_i\),
and hence the cost of evaluating them;
this in turn has a non-negligible impact on the time and space
required to compress and decompress points,
as well as the number of new auxiliary constants that must be stored.

In this section we define a choice of \(L_i\)
reflecting the special symmetry of~\(\KChudnovsky\).
A similar procedure for \(\KGaudry\)
appears in more classical language\footnote{%
    The analogous model of \(\KGaudry\) in~\cite[\S54]{Hudson}
    is called
    ``the equation referred to a Rosenhain tetrad'',
    whose defining equation
    ``...may be deduced from the fact that Kummer's surface
    is the focal surface of the congruence of rays common to a
    tetrahedral complex and a linear complex.''
    Modern cryptographers will understand why we have chosen to give a
    little more algebraic detail here.
}
in~\cite[\S54]{Hudson}.
The trick is to distinguish not one node of \(\KChudnovsky\),
but rather the four nodes
forming the kernel of the \((2,2)\)-isogeny 
\(\Squared\circ\Scaled\circ\Hadamardn : \KChudnovsky \to \KChudnovskydual\),
namely
\begin{align*}
    \pm 0 
    =
    N_0 & = 
    (\squaredn_1:\squaredn_2:\squaredn_3:\squaredn_4)
    \,,
    &
    N_1 & = (\squaredn_2:\squaredn_1:\squaredn_4:\squaredn_3)
    \,,
    \\
    N_2 & = (\squaredn_3:\squaredn_4:\squaredn_1:\squaredn_2)
    \,,
    &
    N_3 & = (\squaredn_4:\squaredn_3:\squaredn_2:\squaredn_1)
    \,.
\end{align*}
We are going to define a coordinate system
where these four nodes become the vertices of a coordinate tetrahedron;
then, projection onto any three of the four coordinates
will represent a projection away from one of these four nodes.
The result will be an isomorphic Kummer \(\KTetra\)
whose defining equation is quadratic in \emph{all four} of its
variables.
This might seem like overkill for point compression---quadratic in just one
variable would suffice---but it has the agreeable effect of dramatically
reducing the overall complexity of the defining equation,
saving time and memory in our compression and decompression algorithms.

The key is the matrix identity
\begin{align}
    \label{eq:alpha-relation}
    \begin{pmatrix}
        \kappad_4 & \kappad_3 & \kappad_2 & \kappad_1
        \\
        \kappad_3 & \kappad_4 & \kappad_1 & \kappad_2
        \\
        \kappad_2 & \kappad_1 & \kappad_4 & \kappad_3
        \\
        \kappad_1 & \kappad_2 & \kappad_3 & \kappad_4
    \end{pmatrix}
    \begin{pmatrix}
        \aan & \bbn & \ccn & \ddn
        \\
        \bbn & \aan & \ddn & \ccn
        \\
        \ccn & \ddn & \aan & \bbn
        \\
        \ddn & \ccn & \bbn & \aan
    \end{pmatrix}
    =
    8\aad\bbd\ccd\ddd 
    \begin{pmatrix}
        0 & 0 & 0 & 1
        \\
        0 & 0 & 1 & 0
        \\
        0 & 1 & 0 & 0
        \\
        1 & 0 & 0 & 0
    \end{pmatrix}
    \,,
\end{align}
which tells us that the projective isomorphism
\(\Tetran\colon\PP^3\to\PP^3\)
defined by
\begin{align*}
    \Tetran 
    \colon
    \left(\begin{array}{r}
        \XCn_1
        \\
        : \XCn_2
        \\
        : \XCn_3
        \\
        : \XCn_4
    \end{array}\right)
    \mapsto 
    \left(\begin{array}{r}
        \XTn_1
        \\
        : \XTn_2
        \\
        : \XTn_3
        \\
        : \XTn_4
    \end{array}\right)
    =
    \left(\begin{array}{r}
        \kappad_4\XCn_1 + \kappad_3\XCn_2 + \kappad_2\XCn_3 + \kappad_1\XCn_4
        \\
        \!: 
        \kappad_3\XCn_1 + \kappad_4\XCn_2 + \kappad_1\XCn_3 + \kappad_2\XCn_4
        \\
        \!: 
        \kappad_2\XCn_1 + \kappad_1\XCn_2 + \kappad_4\XCn_3 + \kappad_3\XCn_4
        \\
        \!: 
        \kappad_1\XCn_1 + \kappad_2\XCn_2 + \kappad_3\XCn_3 + \kappad_4\XCn_4
    \end{array}\right)
\end{align*}
maps the four ``kernel'' nodes
to the corners of a coordinate tetrahedron:
\begin{align*}
    \Tetran(N_0) & = (0:0:0:1) \,,
    &
    \Tetran(N_2) & = (0:1:0:0) \,,
    \\
    \Tetran(N_1) & = (0:0:1:0) \,,
    &
    \Tetran(N_3) & = (1:0:0:0) \,.
\end{align*}
The image of \(\KChudnovsky\)
under \(\Tetran\)
is the \textbf{tetragonal surface}
\[
    \KTetra
    :
    4t\XTn_1\XTn_2\XTn_3\XTn_4
    =
    \begin{array}{c}
        r_1^2(\XTn_1\XTn_2+\XTn_3\XTn_4)^2
        +
        r_2^2(\XTn_1\XTn_3+\XTn_2\XTn_4)^2
        +
        r_3^2(\XTn_1\XTn_4+\XTn_2\XTn_3)^2
        \\
        {} - 2r_1s_1((\XTn_1^2+\XTn_2^2)\XTn_3\XTn_4 + \XTn_1\XTn_2(\XTn_3^2+\XTn_4^2))
        \\
        {} - 2r_2s_2((\XTn_1^2+\XTn_3^2)\XTn_2\XTn_4 + \XTn_1\XTn_3(\XTn_2^2+\XTn_4^2))
        \\
        {} - 2r_3s_3((\XTn_1^2+\XTn_4^2)\XTn_2\XTn_3 + \XTn_1\XTn_4(\XTn_2^2+\XTn_3^2))
    \end{array}
\]
where \(t = 16 \aan\bbn\ccn\ddn \aad\bbd\ccd\ddd\) and
\begin{align*}
    r_1 & = (\aan\ccn-\bbn\ddn)(\aan\ddn-\bbn\ccn)
    \,,
    &
    s_1 & = (\aan\bbn-\ccn\ddn)(\aan\bbn+\ccn\ddn)
    \,,
    \\
    r_2 & = (\aan\bbn-\ccn\ddn)(\aan\ddn-\bbn\ccn)
    \,,
    &
    s_2 & = (\aan\ccn-\bbn\ddn)(\aan\ccn+\bbn\ddn)
    \,,
    \\
    r_3 & = (\aan\bbn-\ccn\ddn)(\aan\ccn-\bbn\ddn)
    \,,
    &
    s_3 & = (\aan\ddn-\bbn\ccn)(\aan\ddn+\bbn\ccn)
    \,.
\end{align*}
As promised,
the defining equation of \(\KTetra\) 
is quadratic in all four of its variables.

For compression we project away from \(\Tetran(\pm0) = (0:0:0:1)\)
onto the \((L_1:L_2:L_3)\)-plane.
Rewriting the defining equation as a quadratic in \(\XTn_4\) 
gives
\[
    \KTetra
    : 
    K_4(\XTn_1,\XTn_2,\XTn_3)
    -
    2K_3(\XTn_1,\XTn_2,\XTn_3)\XTn_4
    + 
    K_2(\XTn_1,\XTn_2,\XTn_3)\XTn_4^2
    = 0
\]
where 
\begin{align*}
    K_2 
	&
	:= 
		r_3^2 \XTn_1^2
		+ r_2^2 \XTn_2^2
		+ r_1^2 \XTn_3^2
        - 2\left(
            r_3s_3 \XTn_2\XTn_3
		    + r_2s_2 \XTn_1\XTn_3
		    + r_1s_1 \XTn_1\XTn_2
        \right)
    \,,
    \\
    K_3 
	&
	:= 
    r_1s_1(\XTn_1^2+\XTn_2^2)\XTn_3
    +
    r_2s_2(\XTn_1^2+\XTn_3^2)\XTn_2
    +
    r_3s_3(\XTn_2^2+\XTn_3^2)\XTn_1
    \\
    & \qquad {}
    + (2t - (r_1^2+r_2^2+r_3^2))\XTn_1\XTn_2\XTn_3
    \,,
    \\
    K_4 
	& 
	:= 
        r_3^2 \XTn_2^2\XTn_3^2 
		+
        r_2^2 \XTn_1^2\XTn_3^2 
		+
        r_1^2 \XTn_1^2\XTn_2^2 
        - 2\left(
            r_3s_3 \XTn_1
            + 
            r_2s_2 \XTn_2
            +
            r_1s_1 \XTn_3
        \right) \XTn_1\XTn_2\XTn_3
    \,.
\end{align*}

\begin{lemma}
    \label{lemma:nonconstant}
    If \((l_1:l_2:l_3:l_4)\) is a point on \(\KTetra\),
    then
    \[
        K_2(l_1,l_2,l_3)
        =
        K_3(l_1,l_2,l_3)
        =
        K_4(l_1,l_2,l_3)
        =
        0
        \iff
        l_1 = l_2 = l_3 = 0
        \,.
    \]
\end{lemma}
\begin{proof}
    Write \(k_i\) for \(K_i(l_1,l_2,l_3)\).
    If \((l_1,l_2,l_3) = 0\) 
    then \((k_2,k_3,k_4) = 0\),
    because each \(K_i\) is nonconstant and homogeneous.
    Conversely,
    if \((k_2,k_3,k_4) = 0\) 
    and \((l_1,l_2,l_3) \not= 0\)
    then we could embed a line in \(\KTetra\)
    via \(\lambda \mapsto (l_1:l_2:l_3:\lambda)\);
    but this is a contradiction,
    because \(\KTetra\) contains no lines. 
    \qed
\end{proof}


\subsection{Compression and decompression for \(\KChudnovsky\)}

In practice, we compress points on \(\KChudnovsky\)
to tuples \((l_1,l_2,\tau,\sigma)\),
where~\(l_1\) and~\(l_2\) are field elements
and \(\tau\) and~\(\sigma\) are bits.
The recipe is
\begin{enumerate}
    \item
        Map \((\XCn_1:\XCn_2:\XCn_3:\XCn_4)\)
        through \(\Tetran\)
        to a point \((\XTn_1:\XTn_2:\XTn_3:\XTn_4)\) on \(\KTetra\).
    \item
        Compute the unique \((l_1,l_2,l_3,l_4)\)
        in one of the forms
        \((*,*,1,*)\), \((*,1,0,*)\), \((1,0,0,*)\), or \((0,0,0,1)\)
        such that \((l_1:l_2:l_3:l_4)=(\XTn_1:\XTn_2:\XTn_3:\XTn_4)\).
    \item
        Compute
        \(k_2 = K_2(l_1,l_2,l_3)\),
        \(k_3 = K_3(l_1,l_2,l_3)\),
        and
        \(k_4 = K_4(l_1,l_2,l_3)\).
    \item
        Define the bit
        \(\sigma = \mathtt{Sign}(k_2l_4-k_3)\);
        then \((l_1,l_2,l_3,\sigma)\)
        determines \(l_4\).
        Indeed,
        \(q(l_4) = 0\),
        where \(q(X) = k_2X^2 - 2k_3X + k_4\);
        and Lemma~\ref{lemma:nonconstant}
        tells us that \(q(X)\) is either quadratic, linear, or
        identically zero.
        \begin{itemize}
            \item
                If \(q\) is a nonsingular quadratic,
                then \(l_4\) is determined by \((l_1,l_2,l_3)\) 
                and \(\sigma\),
                because \(\sigma = \mathtt{Sign}(R)\)
                where \(R\) is the correct square root in the quadratic formula 
                \(l_4 = (k_3 \pm \sqrt{k_3^2-k_2k_4})/k_2\).
            \item
                If \(q\) is singular or linear,
                then \((l_1,l_2,l_3)\) determines
                \(l_4\),
                and \(\sigma\) is redundant.
            \item
                If \(q = 0\)
                then \((l_1,l_2,l_3) = (0,0,0)\),
                so \(l_4 = 1\);
                again, \(\sigma\) is redundant.
        \end{itemize}
        Setting \(\sigma = \mathtt{Sign}(k_2l_4 - k_3)\)
        in every case, regardless of whether or not we need it to
        determine \(l_4\), avoids ambiguity and simplifies code.
    \item
        The normalization in Step 2 forces
        \(l_3\in\{0,1\}\);
        so encode~\(l_3\) as a single bit \(\tau\).
\end{enumerate}
The datum \((l_1,l_2,\tau,\sigma)\) completely determines
\((l_1,l_2,l_3,l_4)\),
and thus determines \((\XCn_1:\XCn_2:\XCn_3:\XCn_4) =
\Tetran^{-1}((l_1:l_2:l_2:l_4))\).
Conversely,
the normalization in Step~2
ensures that
\((l_1,l_2,\tau,\sigma)\)
is uniquely determined by \((\XCn_1:\XCn_2:\XCn_3:\XCn_4)\),
and is independent of the representative values of the~\(\XCn_i\).

Algorithm~\ref{alg:compress}
carries out the compression process above;
the most expensive step is the computation of an inverse in \(\FF_p\).
Algorithm~\ref{alg:decompress}
is the corresponding
decompression algorithm;
its cost is dominated by computing a square root in~\(\FF_p\).

\begin{algorithm}[ht]
    \setstretch{1.1}
    \caption{Kummer point compression for \(\KChudnovsky\)}
    \label{alg:compress}
    \SetKwProg{Function}{function}{}{}
    \SetKwInOut{Cost}{Cost}
    \SetKwData{k}{k}
    \SetKwData{L}{L}
    \SetKwData{l}{l}
    \SetKwData{R}{R}
    \Function{{\tt Compress}}{
        \KwIn{\(\pm P\) in \(\KChudnovsky(\FF_p)\)}
        \KwOut{%
            \((l_1,l_2,\tau,\sigma)\) with \(l_1, l_2 \in \FF_p\) 
            and \(\sigma,\tau \in \{0,1\}\)
        }
        \Cost{%
            \(
            8\MUL +
            5\SQR +
            12\MLC +
            8\ADD +
            5\SUB +
            1\INV
            \)
        }
        %
        \(
            \Big(
                \begin{array}{l}
                    \L_1,\,\L_2,
                    \\
                    \L_3,\,\L_4
                \end{array}
            \Big)
            \gets 
            \Big(
                \begin{array}{l}
                    \Dotprod(\pm P, (\kappad_4,\kappad_3,\kappad_2,\kappad_1)),
                    \,
                    \Dotprod(\pm P, (\kappad_3,\kappad_4,\kappad_1,\kappad_2)),
                    \\
                    \Dotprod(\pm P, (\kappad_2,\kappad_1,\kappad_4,\kappad_3)),
                    \,
                    \Dotprod(\pm P, (\kappad_1,\kappad_2,\kappad_3,\kappad_4))
                \end{array}
            \Big)
        \)
        \label{alg:compress:1}
        \;
        \uIf{\(\L_3 \not= 0\)}{
            \label{alg:compress:2}
            \((\tau,\lambda) \gets (1,\L_3^{-1})\)
            \tcp*{Normalize to \((*:*:1:*)\)}
        }
        \uElseIf{\(\L_2 \not= 0\)}{
            \((\tau,\lambda) \gets (0,\L_2^{-1})\)
            \tcp*{Normalize to \((*:1:0:*)\)}
        }
        \uElseIf{\(\L_1 \not= 0\)}{
            \((\tau,\lambda) \gets (0,\L_1^{-1})\)
            \tcp*{Normalize to \((1:0:0:*)\)}
        }
        \Else{
            \((\tau,\lambda) \gets (0,\L_4^{-1})\)
            \tcp*{Normalize to \((0:0:0:1)\)}
        }
        \(
            (\l_1,\l_2,\l_4) 
            \gets
            (\L_1\cdot\lambda,\L_2\cdot\lambda,\L_4\cdot\lambda)
        \)
        \label{alg:compress:8}
        \tcp*{\((\l_1:\l_2:\tau:\l_4) = (\L_1:\L_2:\L_3:\L_4)\)}
        \((\k_2,\k_3) \gets (K_2(\l_1,\l_2,\tau),K_3(\l_1,\l_2,\tau))\)
        \ifapp
        \tcp*{See Algorithm~\ref{alg:getk2},\ref{alg:getk3}}
        \else
        \;
        \fi
        \(\R \gets \k_2\cdot \l_4 - \k_3\)
        \;
        \(\sigma \gets \mathtt{Sign}(\R)\)
        \;
        \Return{\((\l_1,\l_2,\tau,\sigma)\)}
    }
\end{algorithm}

\begin{algorithm}[ht]
    \setstretch{1.1}
    \caption{Kummer point decompression to \(\KChudnovsky\)}
    \label{alg:decompress}
    \SetKwProg{Function}{function}{}{}
    \SetKwInOut{Cost}{Cost}
    \SetKwData{L}{L}
    \SetKwData{k}{k}
    \SetKwData{R}{R}
    \Function{{\tt Decompress}}{
        \KwIn{%
            \((l_1,l_2,\tau,\sigma)\) with \(l_1, l_2 \in \FF_p\) 
            and \(\tau,\sigma \in \{0,1\}\)
        }
        \KwOut{The point \(\pm P\) in \(\KChudnovsky(\FF_p)\)
            such that \(\texttt{Compress}(\pm P) = (l_1,l_2,\tau,\sigma)\),
            or \(\bot\) if no such \(\pm P\) exists
        }
        \Cost{%
            \(
            10\MUL +
            9\SQR +
            18\MLC +
            13\ADD +
            8\SUB +
            1\SQRT
            \)
        }
        \(
            (\k_2,\k_3,\k_4) 
            \gets 
            (K_2(l_1,l_2,\tau), K_3(l_1,l_2,\tau), K_4(l_1,l_2,\tau))
        \)
        \ifapp
        \tcp*{Alg.~\ref{alg:getk2},\ref{alg:getk3},\ref{alg:getk4}}
        \else
        \;
        \fi
        \uIf{\(\k_2 = 0\) {\rm\bf and} \(\k_3 = 0\)}{
            \If{\((l_1,l_2,\tau,\sigma) \not= (0,0,0,\mathtt{Sign}(0))\)}{%
                \Return{\(\bot\)}
                \tcp*{Invalid compression}
            }
            \(\L \gets (0,0,0,1)\)
        }
        \uElseIf{\(\k_2 = 0\) {\rm\bf and} \(\k_3 \not= 0\)}{
            \If{\(\sigma \not= \mathtt{Sign}(-\k_3)\)}{%
                \Return{\(\bot\)}
                \tcp*{Invalid compression}
            }
            \(
                \L
                \gets
                (2\cdot l_1\cdot \k_3, 2\cdot l_2\cdot \k_3, 2\cdot \tau \cdot \k_3, \k_4)
            \)
            \tcp*{\(\k_4 = 2\k_3l_4\)}
        }
        \Else{
            \(\Delta \gets \k_3^2 - \k_2\k_4\)
            \;
            \(\R \gets \mathtt{HasSquareRoot}(\Delta,\sigma)\) 
            \tcp*{%
                \(\R = \bot\) or \(\R^2 = \Delta\), \(\mathtt{Sign}(\R) = \sigma\)
            }
            \If{\(\R = \bot\)}{%
                \Return{\(\bot\)}
                \tcp*{No preimage in \(\KTetra(\FF_p)\)}
            }
            \(
                \L
                \gets
                (\k_2\cdot l_1, \k_2\cdot l_2, \k_2\cdot \tau, \k_3 + \R)
            \)
            \tcp*{\(\k_3 + \R = \k_2l_4\)}
        }
        \(
            \Big(
                \begin{array}{l}
                    \XCn_1,\,\XCn_2,
                    \\
                    \XCn_3,\,\XCn_4
                \end{array}
            \Big)
            \gets
            \Big(
                \begin{array}{l}
                    \texttt{Dot}(
                        \L,
                        (\squaredn_4, \squaredn_3, \squaredn_2, \squaredn_1)
                    ),\,
                    \texttt{Dot}(
                        \L,
                        (\squaredn_3, \squaredn_4, \squaredn_1, \squaredn_2)
                    ),
                    \\
                    \texttt{Dot}(
                        \L,
                        (\squaredn_2, \squaredn_1, \squaredn_4, \squaredn_3)
                    ),\,
                    \texttt{Dot}(
                        \L,
                        (\squaredn_1, \squaredn_2, \squaredn_3, \squaredn_4)
                    )
                \end{array}
            \Big)
        \)
        \label{alg:decompress:T-inv}
        \;
        \Return{\((\XCn_1:\XCn_2:\XCn_3:\XCn_4)\)}
    }
\end{algorithm}

\begin{proposition}
    Algorithms~\ref{alg:compress}
    and~\ref{alg:decompress}
    ({\rm \tt Compress} and {\rm \tt Decompress})
    satisfy the following properties:
    given \((l_1,l_2,\tau,\sigma)\) in \(\FF_p^2\times\{0,1\}^2\),
    {\rm \tt Decompress}
    always returns either a valid point in \(\KChudnovsky(\FF_p)\)
    or \(\bot\);
    and 
    for every \(\pm P\)  in \(\KChudnovsky(\FF_p)\),
    we have
    \[
        {\rm \tt Decompress}(
            {\rm \tt Compress}(
                \pm P 
            )
        ) 
        = 
        \pm P 
        \,.
    \]
\end{proposition}
\begin{proof}
    In Algorithm~\ref{alg:decompress}
    we are given \((l_1,l_2,\tau,\sigma)\).
    We can immediately set \(l_3 = \tau\),
    viewed as an element of \(\FF_p\).
    We want to compute an \(l_4\) in \(\FF_p\),
    if it exists,
    such that \(k_2l_4^2 - 2k_3l_4 + k_4 = 0\)
    and \(\mathtt{Sign}(k_2l_4 - l_3) = \sigma\)
    where \(k_i = K_i(l_1,l_2,l_3)\).
    If such an \(l_4\) exists,
    then we will have a preimage \((l_1:l_2:l_3:l_4)\) in
    \(\KTetra(\FF_p)\),
    and we can return the decompressed
    \(\Tetran^{-1}((l_1:l_2:l_3:l_4))\) in \(\KChudnovsky\).

    If \((k_2,k_3) = (0,0)\) 
    then \(k_4 = 2k_3l_4-k_2l_4^2 = 0\),
    so \(l_1 = l_2 = \tau = 0\)
    by Lemma~\ref{lemma:nonconstant}.
    The only legitimate datum in this form is
    is \((l_1:l_2:\tau:\sigma) = (0:0:0:\mathtt{Sign}(0))\).
    If this was the input,
    then the preimage is \((0:0:0:1)\);
    otherwise we return \(\bot\).
    
    If \(k_2 = 0\) but \(k_3 \not= 0\),
    then 
    \(k_4 = 2k_3l_4\),
    so \((l_1:l_2:\tau:l_4) = (2k_3l_1:2k_3l_2:2k_3\tau:k_4)\).
    The datum is a valid compression unless
    \(\sigma \not= \mathtt{Sign}(-k_3)\),
    in which case we return \(\bot\);
    otherwise,
    the preimage is
    \((2k_3l_1:2k_3l_2:2k_3\tau:k_4)\).

    If \(k_2 \not= 0\),
    then the quadratic formula tells us that
    any preimage satisfies
    \(k_2l_4 = k_3 \pm \sqrt{k_3^2 - k_2k_4}\),
    with the sign determined by \(\mathtt{Sign}(k_2l_4-k_3)\).
    If \(k_3^2 - k_2k_4\) is not a square in \(\FF_p\)
    then there is no such \(l_4\) in \(\FF_p\);
    the input is illegitimate, so we return \(\bot\).
    Otherwise, we have a preimage
    \(
        (k_2l_1:k_2l_2:k_2l_3:l_3\pm\sqrt{k_3^2-k_2k_4})
    \).

    Line~\ref{alg:decompress:T-inv}
    maps the preimage \((l_1:l_2:l_3:l_4)\) in \(\KTetra(\FF_p)\)
    back to \(\KChudnovsky(\FF_p)\)
    via \(\Tetran^{-1}\), 
    yielding the decompressed point \((\XCn_1:\XCn_2:\XCn_3:\XCn_4)\).
    \qed
\end{proof}

\subsection{Using cryptographic parameters}

Our compression scheme works out particularly nicely 
for the Gaudry--Schost Kummer over \(\FF_{2^{127}-1}\).
First, since every field element fits into 127 bits,
every compressed point fits into exactly 256 bits.
Second,
the auxiliary constants are small:
we have
\(
    (\kappad_1:\kappad_2:\kappad_3:\kappad_4) 
    =
    ( -961 : 128 : 569 : 1097 )
\),
each of which fits into well under 16 bits.
Computing the polynomials \(K_2\), \(K_3\), \(K_4\)
and dividing them all through by \(11^2\) 
(which does not change the roots of the quadratic)
gives
\begin{align}
    K_2(l_1,l_2,\tau)
    & =
    (q_5 l_1)^2 + (q_3 l_2)^2 + (q_4\tau)^2
    - 2q_3\big(
        q_2 l_1l_2 + \tau(q_0 l_1 - q_1 l_2)
    \big)
    \,,
    \label{eq:k2}\\
    K_3(l_1,l_2,\tau)
    & =
    q_3\big(
        q_0(l_1^2 + \tau)l_2 - q_1 l_1(l_2^2 + \tau) +
        q_2(l_1^2+l_2^2)\tau
    \big)
    - q_6q_7 l_1l_2\tau
    \,,
    \label{eq:k3}\\
    K_4(l_1,l_2,\tau)
    & =
    (
        (q_3 l_1)^2
        + (q_5 l_2)^2
        - 2q_3l_1l_2\big(
            q_0 l_2 - q_1l_1 + q_2
        \big)
    )\tau
    + (q_4 l_1l_2)^2
    \,,\label{eq:k4}%
\end{align}
where
\(
    (q_0,\ldots,q_7)
    =
    (3575,
    9625,
    4625,
    12259,
    11275,
    7475,
    6009,
    43991
    )
\);
each of the \(q_i\) fits into 16 bits.
In total, the twelve new constants we need for 
{\tt Compress} and {\tt Decompress}
together fit into less than two field elements' worth of space.

%

%% file: implementation.tex
\def\kummeravrcodesize{\(17\,880\)\xspace}
\def\kummeravrsigncycles{\(10\,477\,347\)\xspace}
\def\kummeravrverifycycles{\(20\,423\,937\)\xspace}
\def\kummeravrsignstack{\(417\)\xspace}
\def\kummeravrverifystack{\(609\)\xspace}

\def\kummerarmcodesize{\(18\,064\)\xspace}
\def\kummerarmsigncycles{\(2\,908\,215\)\xspace}
\def\kummerarmverifycycles{\(5\,694\,414\)\xspace}
\def\kummerarmsignstack{\(580\)\xspace}
\def\kummerarmverifystack{\(808\)\xspace}

\def\ellipticavrcodesize{\(21\,347\)\xspace}
\def\ellipticavrsigncycles{\(14\,067\,995\)\xspace}
\def\ellipticavrverifycycles{\(25\,355\,140\)\xspace}
\def\ellipticavrsignstack{\(512\)\xspace}
\def\ellipticavrverifystack{\(644\)\xspace}

\def\ellipticarmcodesize{\(18\,443\)\xspace}
\def\ellipticarmsigncycles{\(3\,889\,116\)\xspace}
\def\ellipticarmverifycycles{\(6\,793\,695\)\xspace}
\def\ellipticarmsignstack{\(660\)\xspace}
\def\ellipticarmverifystack{\(788\)\xspace}


\def\kummeravrladder{\(9\,624\,637\)\xspace}
\def\kummeravrcheck{\(84\,424\)\xspace}
\def\kummeravrcompress{\(212\,374\)\xspace}
\def\kummeravrdecompress{\(211\,428\)\xspace}

\def\kummerarmladder{\(2\,683\,371\)\xspace}
\def\kummerarmcheck{\(24\,249\)\xspace}
\def\kummerarmcompress{\(62\,165\)\xspace}
\def\kummerarmdecompress{\(62\,471\)\xspace}

\def\ellipticavrladder{\(12\,539\,098\)\xspace}
\def\ellipticavrcheck{\(46\,546\)\xspace}
\def\ellipticavrcompress{\(1\,067\,004\)\xspace}
\def\ellipticavrdecompress{\(694\)\xspace}

\def\ellipticarmladder{\(3\,338\,554\)\xspace}
\def\ellipticarmcheck{\(17\,044\)\xspace}
\def\ellipticarmcompress{\(270\,867\)\xspace}
\def\ellipticarmdecompress{\(102\)\xspace}

\section{Implementation}\label{sec:imp}

In this section we present the results of the implementation of the scheme on the
AVR ATmega and ARM Cortex M0 platforms.
We have a total of four implementations:
on both platforms we implemented
both the Curve25519-based scheme and 
the scheme based on a fast Kummer surface in genus 2.
The benchmarks for the AVR software are obtained from
the Arduino MEGA development board containing
an ATmega2560 MCU, compiled with GCC v4.8.1.
For the Cortex M0, they
are measured on the STM32F051R8 MCU on the
STMF0Discovery board, compiled with Clang v3.5.0.
We refer to the (publicly available) code for more detailed compiler settings.
For both Diffie--Hellman and signatures we follow the eBACS~\cite{ebacs} API.


\subsection{Core functionality}

The arithmetic of the underlying finite fields is well-studied and optimized,
and we do not reinvent the wheel.
For field arithmetic in \(\F_{2^{255}-19}\)
we use the highly optimized 
functions presented by Hutter and Schwabe~\cite{HS13} for the AVR ATmega,
and
the code from D\"{u}ll et al.~\cite{DHH+15} for the Cortex M0.
For arithmetic in \(\F_{2^{127}-1}\) 
we use the functions from Renes et al.~\cite{RSSB16},
which in turn rely on~\cite{HS13} for the AVR ATmega,
and on~\cite{DHH+15} for the Cortex M0.

The {\tt SHAKE128} functions for the ATmega are taken from~\cite{Keccak}, while
on the Cortex M0 we use a modified version from~\cite{AJS16}.
Cycle counts for the main functions defined in the rest of this paper
are presented in Table~\ref{tab:counts}.
Notably, the {\tt Ladder} routine is by far the most expensive function.
In genus 1 the {\tt Compress} function is relatively costly 
(it is essentially an inversion), while in genus 2
{\tt Check}, {\tt Compress} and {\tt Decompress}
have only minor impact on the total run-time.
More interestingly, as seen in Table~\ref{tab:avrsig} and Table~\ref{tab:armsig},
the simplicity of operating only on the Kummer variety allows 
		 smaller code and less stack usage.

\begin{table}[ht]
\centering
\renewcommand{\tabcolsep}{0.1cm}
\renewcommand{\arraystretch}{1.1}
	\begin{tabular}{|c|c|c|r|r|}
\hline
{\bf Genus} & 
{\bf Function} & 
{\bf Ref.} & 
{\bf AVR ATmega} & 
{\bf ARM Cortex M0} \\
\hline
\hline
\multirow{4}{*}{1}
& {\tt Ladder}
& \ifapp Alg.~\ref{alg:montladder}
    \else -- \fi
& \ellipticavrladder
& \ellipticarmladder
\\
& {\tt Check}
& Alg.~\ref{alg:mont-check}
& \ellipticavrcheck
&	\ellipticarmcheck
\\
& {\tt Compress}
& \S\ref{subsec:ecmont}
& \ellipticavrcompress
& \ellipticarmcompress
\\
& {\tt Decompress}
& \S\ref{subsec:ecmont}
& \ellipticavrdecompress
& \ellipticarmdecompress
\\
\hline
\hline
\multirow{4}{*}{2}
& {\tt Ladder}
& \ifapp Alg.~\ref{alg:ladder-kummer}
    \else -- \fi
& \kummeravrladder
& \kummerarmladder
\\
& {\tt Check}\footnotemark
& Alg.~\ref{alg:kummer-check}
& \kummeravrcheck
& \kummerarmcheck
\\
& {\tt Compress}
& Alg.~\ref{alg:compress}
& \kummeravrcompress
& \kummerarmcompress
\\
& {\tt Decompress}
& Alg.~\ref{alg:decompress}
& \kummeravrdecompress
& \kummerarmdecompress
\\
\hline
\end{tabular}
\vspace{0.2cm}
	\caption{
		Cycle counts for the four key functions of \xsign
		at the 128-bit security level.
	}
	\label{tab:counts}
\end{table}
\footnotetext{The implementation decompresses \(\pm\R\)
	within {\tt Check}, while Algorithm~\ref{alg:kummer-check} assumes
		\(\pm\R\) to be decompressed.
We have subtracted the cost of the {\tt Decompress} function once.}

\subsection{Comparison to previous work}

There are not many implementations of complete signature
and key exchange schemes on microcontrollers.
On the other hand, there are implementations of scalar
multiplication on elliptic curves.
The current fastest on our platforms are 
presented by D\"{u}ll et al.~\cite{DHH+15},
and since we are relying on exactly the same 
arithmetic, we have essentially the same
results.
Similarly, the current records for 
scalar multiplication on Kummer surfaces
are presented by Renes et al.~\cite{RSSB16}.
Since we use the same underlying functions, 
we have similar results.

More interestingly, we compare the speed and memory
usage of signing and verification to best known
results of implementations of complete signature schemes.
To the best of our knowledge, the only other works are
the Ed25519-based scheme by Nascimento et al~\cite{NLD15},
the Four\(\mathbb{Q}\)-based scheme 
(obtaining fast scalar multiplication by relying on easily computable
endomorphisms)
by Liu et al~\cite{LLP+17},
and the genus-2 implementation from~\cite{RSSB16}.

\paragraph{AVR ATmega.} As we see in Table~\ref{tab:avrsig}, our implementation
of the scheme based on Curve25519 outperforms the Ed25519-based scheme
from~\cite{NLD15} in every way.
It reduces the number of clock cycles needed for {\tt sign} resp. {\tt verify}
by more than 26\% resp. 17\%, while reducing
stack usage by more than 65\% resp. 47\%.
Code size is not reported in~\cite{NLD15}.
Comparing against the Four\(\mathbb{Q}\) implementation
of~\cite{LLP+17}, we see a clear trade-off
between speed and size:
Four\(\mathbb{Q}\) has a clear speed advantage,
but \xsign on Curve25519 requires only a fraction of the stack space.

The implementation based on the Kummer surface of the
genus-2 Gaudry--Schost Jacobian does better than the 
Curve25519-based implementation across the board.
Compared to~\cite{RSSB16},
the stack usage of {\tt sign} resp. {\tt verify}
decreases by more than 54\% resp. 38\%, while decreasing code size
by about 11\%. On the other hand, verification is
about 26\% slower. This is explained by the fact
that in~\cite{RSSB16} the signature is compressed to
48 bytes (following Schnorr's suggestion),
which means that one of the scalar multiplications
in verification is only half length.
Comparing to the Four\(\mathbb{Q}\) implementation of~\cite{LLP+17},
again we see a clear trade-off
between speed and size,
but this time the loss of speed is less pronounced
than in the comparison with Curve25519-based \xsign.

\begin{table}[ht]
\centering
\renewcommand{\tabcolsep}{0.1cm}
\renewcommand{\arraystretch}{1.1}
	\begin{tabular}{|c|c|c|r|r|r|}
\hline
{\bf Ref.} & 
{\bf Object} & 
{\bf Function} & 
{\bf Clock cycles} & 
{\bf Stack} &
{\bf Code size\footnotemark} 
\\
\hline
\hline
\multirow{2}{*}{\cite{NLD15}} &
\multirow{2}{*}{Ed25519} &
{\tt sign} &
\(19\,047\,706\)\xspace &
\(1\,473\)\xspace bytes &
\multirow{2}{*}{--}
\\
& 
& 
{\tt verify} &
\(30\,776\,942\)\xspace &
\(1\,226\)\xspace bytes &
\\
\hline
\multirow{2}{*}{\cite{LLP+17}} &
\multirow{2}{*}{Four\(\mathbb{\Q}\)} &
{\tt sign} &
\(5\,174\,800\)\xspace &
\(1\,572\)\xspace bytes &
25\,354 bytes
\\
& 
& 
{\tt verify} &
\(11\,003\,800\)\xspace &
\(4\,957\)\xspace bytes &
33\,372 bytes
\\
\hline
\multirow{2}{*}{{\bf This work}} &
\multirow{2}{*}{Curve25519} &
{\tt sign} &
\ellipticavrsigncycles &
\ellipticavrsignstack bytes &
\multirow{2}{*}{\ellipticavrcodesize bytes}
\\
& 
& 
{\tt verify} &
\ellipticavrverifycycles &
\ellipticavrverifystack bytes &
\\
\hline
\hline
\multirow{2}{*}{\cite{RSSB16}} &
Gaudry-- &
{\tt sign} &
\(10\,404\,033\)\xspace &
\(926\)\xspace bytes &
\multirow{2}{*}{\(20\,242\)\xspace bytes}
\\
&
Schost \(\Jac\)
& 
{\tt verify} &
\(16\,240\,510\)\xspace &
\(992\)\xspace bytes &
\\
\hline
\multirow{2}{*}{{\bf This work}} &
Gaudry-- &
{\tt sign} &
\kummeravrsigncycles &
\kummeravrsignstack bytes &
\multirow{2}{*}{\kummeravrcodesize bytes} 
\\
& 
Schost \(\Kum\)
& 
{\tt verify} &
\kummeravrverifycycles &
\kummeravrverifystack bytes &
\\
\hline
\end{tabular}
\vspace{0.2cm}
	\caption{
		Performance comparison of the \xsign signature scheme against 
			the current best implementations, on the AVR ATmega platform.
	}
	\label{tab:avrsig}
\end{table}
\footnotetext{%
All reported code sizes except those from~\cite[Table 6]{LLP+17}
include support for both signatures and key exchange.}

\paragraph{ARM Cortex M0.} 
In this case there is no elliptic-curve-based signature
scheme to compare to, so we present the first.
As we see in Table~\ref{tab:armsig},
it is significantly slower than its genus-2 counterpart
in this paper
(as should be expected), 
while using a similar amount of stack and code.
The genus-2 signature scheme has similar trade-offs on
this platform when compared to the implementation by
Renes et al.~\cite{RSSB16}.
The stack usage for {\tt sign} resp. {\tt verify} is reduced
by about 57\% resp. 43\%, while code size is reduced
by about 8\%. For the same reasons as above, 
verification is about 28\% slower.

\begin{table}[ht]
\centering
\renewcommand{\tabcolsep}{0.1cm}
\renewcommand{\arraystretch}{1.1}
	\begin{tabular}{|c|c|c|r|r|r|}
\hline
{\bf Ref.} & 
{\bf Object} & 
{\bf Function} & 
{\bf Clock cycles} & 
{\bf Stack} &
{\bf Code size}\footnotemark
\\
\hline
\hline
\multirow{2}{*}{{\bf This work}} &
\multirow{2}{*}{Curve25519} &
{\tt sign} &
\ellipticarmsigncycles &
\ellipticarmsignstack bytes &
\multirow{2}{*}{\ellipticarmcodesize bytes}
\\
& 
& 
{\tt verify} &
\ellipticarmverifycycles &
\ellipticarmverifystack bytes &
\\
\hline
\hline
\multirow{2}{*}{\cite{RSSB16}} &
        Gaudry-- &
{\tt sign} &
\(2\,865\,351\)\xspace &
\(1\,360\)\xspace bytes &
\multirow{2}{*}{\(19\,606\)\xspace bytes}
\\
& 
Schost \(\Jac\) & 
{\tt verify} &
\(4\,453\,978\)\xspace &
\(1\,432\)\xspace bytes &
\\
\hline
\multirow{2}{*}{{\bf This work}} &
Gaudry-- &
{\tt sign} &
\kummerarmsigncycles &
\kummerarmsignstack bytes &
\multirow{2}{*}{\kummerarmcodesize bytes} 
\\
& 
Schost \(\Kum\) & 
{\tt verify} &
\kummerarmverifycycles &
\kummerarmverifystack bytes &
\\
\hline
\end{tabular}
\vspace{0.2cm}
	\caption{
		Performance comparison of the \xsign signature scheme against 
			the current best implementations, on the ARM Cortex M0 platform.
			}
	\label{tab:armsig}
\end{table}
\footnotetext{In this work 8\,448 bytes come from the {\tt SHAKE128}
implementation, while~\cite{RSSB16} uses 6\,938 bytes.
One could probably reduce this
significantly by optimizing the implementation, or by
using a more memory-friendly hash function.}

%% file: elliptic-impl.tex
\section{
    Elliptic implementation details
}
\label{sec:elliptic-impl}

The algorithms in this section
complete the description of elliptic \xsign in~\S\ref{sec:elliptic}.

\subsection{Pseudoscalar multiplication}

The {\tt keypair}, {\tt sign}, and {\tt verify}
functions all require {\tt Ladder},
which we define below.
Algorithm~\ref{alg:montladder}
describes the scalar pseudomultiplication
that we implemented for Montgomery curves,
closely following our C reference implementation.
To make our {\tt Ladder} constant-time,
we use a \emph{conditional swap} procedure {\tt CSWAP}.
This takes a single bit and a pair of items as arguments,
and swaps those items if and only if the bit is \(1\).

\begin{algorithm}[ht]
		\setstretch{1.1} 
    \caption{\texttt{Ladder}:
        the Montgomery ladder for elliptic pseudo-multiplication on \(\PP^1\),
        using a combined differential double-and-add 
        (Algorithm~\ref{alg:xDBLADD-mont}).
    }
    \label{alg:montladder}
    \SetKwProg{Function}{function}{}{}
    \SetKwInOut{Cost}{Cost}
    \SetKwData{bit}{{\tt bit}}
    \SetKwData{prevbit}{{\tt prevbit}}
    \SetKwData{swap}{{\tt swap}}
    \SetKwData{x}{\(x\)}
    \SetKwData{U}{U}
    \SetKwData{V}{V}
    \SetKwData{W}{W}
	\Function{{\tt Ladder}}{
        \KwIn{%
            \(m = \sum_{i=0}^{255}m_i2^i \in \ZZ\)
            and
            \(\pm P = (x:1)\in\PP^1(\FF_p)\)\,, \(x\neq0\)
        }
        \KwOut{%
            \(\pm[m]P\)
        }
        \Cost{%
            \(
            1280\MUL +
            1024\SQR + 
            256\MLC + 
            1024\ADD +
            1024\SUB
            \)
        }
        \( \prevbit \gets 0 \)
        \;
        \(
            (\V_0,\V_1) 
            \gets 
            \big(
                ( 1 : 0 )
                ,
                \pm P
            \big)
        \)
        \;
        \For{\(i = 255\) {\rm\bf down to} \(0\)}{%
            \( 
                (\bit, \prevbit, \swap) 
                \gets 
                (m_i, \bit, \bit \oplus \prevbit) 
            \)
            \;
            \( {\tt CSWAP}(\swap, (\V_0, \V_1)) \)
            \;
            \( \xDBLADD(\V_0, \V_1, \x) \)
            \;
        }
        \( {\tt CSWAP}(\bit, (\V_0, \V_1)) \) \;
        \Return{\(\V_0\)}
	}
\end{algorithm}

Algorithm~\ref{alg:xDBLADD-mont}
implements {\tt xDBLADD} for Montgomery curves
in the usual way.
Note that the assumption that \(\pm(\Pc-\Q)\not\in\{(1 : 0), (0 : 1)\}\)
implies that {\tt xDBLADD} will always return the correct result.

\begin{algorithm}
    \setstretch{1.1}
    \caption{{\tt xDBLADD}:
        combined pseudo-addition and doubling on \(\PP^1\).
    }
    \label{alg:xDBLADD-mont}
    \SetKwProg{Function}{function}{}{}
    \SetKwInOut{Cost}{Cost}
    \SetKwData{U}{U}
    \SetKwData{V}{V}
    \SetKwData{W}{W}
    \Function{{\tt xDBLADD}}{
        \KwIn{%
            \(\pm\Pc = (X^P : Z^P)\)
            and
            \(\pm\Q = (X^Q : Z^Q)\)
            in \(\PP^1(\FF_q)\),
            and \(x \in \FF_q^*\)
            such that \((x:1) = \pm(\Pc-\Q)\)
        }
        \KwOut{%
            \((\pm [2]\Pc,\pm(\Pc+\Q))\)
        }
        \Cost{%
            \(
            5\MUL +
            4\SQR + 
            1\MLC + 
            4\ADD +
            4\SUB
            \)
        }
        \begin{multicols}{2}
            \((\U_0,\U_1,\V_0,\V_1) \gets (X^P,Z^P,X^Q,Z^Q)\)
            \;
            \((\W_0,\W_1) \gets (\U_0+\U_1,\U_0-\U_1)\)
            \;
            \((\U_0,\U_1) \gets (\V_0+\V_1,\V_0-\V_1)\)
            \;	
            \((\V_0,\U_1) \gets (\W_0\cdot\U_1, \W_1\cdot\U_0) \)
            \;
            \((\U_0,\V_1) \gets (\V_0+\U_1,\V_0-\U_1)\)
            \;
            \((\U_0,\V_0,\V_1) \gets (\U_0^2,\V_0^2,x\cdot\U_0)\)
            \;
            \((\W_0,\U_0) \gets (\W_1^2,\W_0^2)\)
            \;
            \(\U_1 \gets \U_0 - \W_0\)
            \;
            \(\U_0 \gets \W_0\cdot\U_0\)
            \;
            \(\W_1 \gets \frac{A+2}{4}\cdot\U_1\)
            \;
            \(\W_1 \gets \W_0\cdot\W_1\)
            \;
            \(\U_1 \gets \W_1\cdot\U_1\)
            \;
            \Return{\(\big((\U_0,\U_1),(\V_0,\V_1)\big)\)}
        \end{multicols}
    }
\end{algorithm}

\subsection{The {\tt BValues} subroutine for signature verification}

The elliptic version of the crucial {\tt Check} subroutine of {\tt verify} 
(Algorithm~\ref{alg:mont-check})
used a function {\tt BValues}
to calculate the values of the biquadratic forms
\(B_{XX}\), \(B_{XZ}\), and \(B_{ZZ}\).
This function can be implemented in a number of ways,
with different optimizations for speed or stack usage.
Algorithm~\ref{alg:bvalues}
illustrates the approach we used for {\tt BValues},
motivated by simplicity and stack minimisation.

\begin{algorithm}[ht]
    \setstretch{1.1}
    \caption{        
        \({\tt BValues}\):
        evaluates 
        \(B_{XX}\), \(B_{XZ}\), and \(B_{ZZ}\)
        on \(\PP^1\).
    }
    \label{alg:bvalues}
    \SetKwProg{Function}{function}{}{}
    \SetKwInOut{Cost}{Cost}
    \SetKwData{t}{T}
    \SetKwData{U}{U}
    \SetKwData{V}{V}
    \SetKwData{W}{W}
	\Function{\({\tt BValues}\)}{
        \KwIn{%
            \(\pm \Pc=(X^P:Z^P)\), \(\pm \Q=(X^Q:Z^Q)\) in \(\K(\FF_p)\)
        }
        \KwOut{%
            \(\left(B_{XX}(\pm P,\pm Q),B_{XZ}(\pm P,\pm Q),B_{ZZ}(\pm
            P, \pm Q)\right)\) in \(\FF_p^3\)
        }
        \Cost{%
            \(
            6\MUL +
            2\SQR + 
            1\MLC + 
            7\ADD +
            3\SUB
            \)
        }
        \begin{multicols}{2}
        \( (\t_0,\t_1) \gets (X^P \cdot X^Q, Z^P \cdot Z^Q) \)
        \;
        \( \U \gets (\t_0 - \t_1)^2 \)
        \;
        \( \t_0 \gets \t_0 + \t_1 \)
        \;
        \( (\t_1,\t_2) \gets (X^P \cdot Z^Q, X^Q \cdot Z^P) \)
        \;
        \( \W \gets (\t_1 - \t_2)^2 \)
        \;
        \( \V \gets \t_0\cdot(\t_1 + \t_2) \)
        \;
        \( \t_0 \gets 4 \cdot \t_1 \cdot \t_2 \)
        \;
        \( \t_1 \gets 2 \cdot \t_0 \)
        \;
        \( \t_1 \gets \frac{A+2}{4} \cdot \t_1 \)
        \;
        \( \V \gets \V + \t_1 - \t_0 \)
        \;
        \Return{\((\U,\V,\W)\)}
        \end{multicols}
	}
\end{algorithm}

%% file: kummer-impl.tex
\section{
    Kummer surface implementation details
}
\label{sec:kummer-impl}

The algorithms in this section 
complete the description of Kummer \xsign
in~\S\S\ref{sec:kummer-arithmetic}-\ref{sec:comp}.
They follow our C reference implementation very closely.
Recall that we have the following subroutines:
\begin{itemize}
    \item \(\PPmul\) implements a 4-way parallel multiplication.
        It takes a pair of vectors 
        \((x_1,x_2,x_3,x_4)\)
        and
        \((y_1,y_2,y_3,y_4)\)
        in \(\FF_p^4\),
        and returns \((x_1y_1,x_2y_2,x_3y_3,x_4y_4)\).
    \item
        \(\PPsqr\) implements a 4-way parallel squaring.
        Given a vector 
        \((x_1,x_2,x_3,x_4)\) in \(\FF_p^4\),
        it returns \((x_1^2,x_2^2,x_3^2,x_4^2)\).
    \item
        \(\Hadamard\) implements a Hadamard transform.
        Given a vector 
        \((x_1,x_2,x_3,x_4)\) in \(\FF_p^4\),
        it returns 
        \(
            (
            x_1 + x_2 + x_3 + x_4,
            x_1 + x_2 - x_3 - x_4,
            x_1 - x_2 + x_3 - x_4,
            x_1 - x_2 - x_3 + x_4
            )
        \).
    \item
        \(\Dotprod\) computes the sum of a 4-way multiplication.
        Given a pair of vectors 
        \((x_1,x_2,x_3,x_4)\)
        and
        \((y_1,y_2,y_3,y_4)\)
        in \(\FF_p^4\),
        it returns \(x_1y_1+x_2y_2+x_3y_3+x_4y_4\).
\end{itemize}

\subsection{Scalar pseudomultiplication}

The Montgomery {\tt Ladder} for scalar pseudomultiplication on 
\(\KChudnovsky\) 
is implemented in Algorithm~\ref{alg:ladder-kummer},
replicating the approach in~\cite{RSSB16}.
It relies on the {\tt WRAP} and {\tt xDBLADD} functions,
implemented in Algorithm~\ref{alg:WRAP} respectively~\ref{alg:xDBLADD-kummer}.
The function {\tt WRAP}
takes a Kummer point \(\pm P\) in \(\KChudnovsky(\FF_p)\)
and returns \(w_2\), \(w_3\), and \(w_4\) in \(\FF_p\)
such that
\(
    (1:w_2:w_3:w_4) =
    (1/\XCn_1^P:1/\XCn_2^P:1/\XCn_3^P:1/\XCn_4^P)
\).
The resulting values are required in every {\tt xDBLADD} within {\tt
Ladder};
the idea is to compute them once with a single inversion
at the start of the procedure,
thus avoiding further expensive inversions.
We note that this ``wrapped'' form of the point \(\pm P\)
was previously used as a compressed form for Kummer point transmission,
but since it requires three full field values
it is far from an optimal compression.


\begin{algorithm}[ht]
    \setstretch{1.0}
    \caption{\texttt{Ladder}:
        the Montgomery ladder for pseudomultiplication on
        \(\KChudnovsky\),
        based on a combined differential double-and-add
        (Algorithm~\ref{alg:xDBLADD-kummer}).
    }
    \label{alg:ladder-kummer}
    \SetKwProg{Function}{function}{}{}
    \SetKwInOut{Cost}{Cost}
    \SetKwData{bit}{{\tt bit}}
    \SetKwData{prevbit}{{\tt prevbit}}
    \SetKwData{swap}{{\tt swap}}
    \SetKwData{V}{V}
    \SetKwData{W}{W}
	\Function{{\tt Ladder}}{
        \KwIn{%
            \(m = \sum_{i=0}^{255}m_i2^i \in \ZZ\)
            and
            \(\pm P\in\KChudnovsky(\FF_p)\)
        }
        \KwOut{%
            \(\pm[m]P\)
        }
        \Cost{%
            \(
            1799\MUL +
            3072\SQR + 
            3072\MLC + 
            4096\ADD +
            4096\SUB +
            1\INV
            \)
        }
        \( \prevbit \gets 0\)
        \;
        \( \W \gets \texttt{WRAP}(\pm P) \)
        \;
        \(
            (\V_0,\V_1) 
            \gets 
            \big(
                (\squaredn_1:\squaredn_2:\squaredn_3:\squaredn_4)
                ,
                \pm P
            \big)
        \)
        \;
        \For{\(i = 255\) {\rm\bf down to} \(0\)}{%
            \( (\bit,\prevbit,\swap) \gets (m_i,\bit,\bit \oplus \prevbit) \)
            \;
            \( {\tt CSWAP}(\swap, (\V_0, \V_1)) \)
            \;
            \( \xDBLADD(\V_0, \V_1, \W) \)
            \;
        }
        \( {\tt CSWAP}(\bit, (\V_0, \V_1)) \) \;
        \Return{\(\V_0\)}
	}
\end{algorithm}


\begin{algorithm}[ht]
    \setstretch{1.1}
    \caption{{\tt WRAP}: (pre)computes inverted Kummer point coordinates.}
    \label{alg:WRAP}
    \SetKwProg{Function}{function}{}{}
    \SetKwInOut{Cost}{Cost}
    \SetKwData{V}{V}
    \SetKwData{W}{W}
    \Function{{\tt WRAP}}{
        \KwIn{%
            \(\pm P\in\KChudnovsky(\FF_p)\)
        }
        \KwOut{%
            \((w_2,w_3,w_4)\in \FF_p^3\)
            such that 
            \(
                (1:w_2:w_3:w_4) 
                =
                (1/\XCn_1^P:1/\XCn_2^P:1/\XCn_3^P:1/\XCn_4^P)
            \)
        }
        \Cost{%
            \(
            7\MUL +
            1\INV
            \)
        }
		\(\V_1 \gets \XCn_2^P\cdot \XCn_3^P\)
		\tcp*{1M}
		\(\V_2 \gets \XCn_1^P/(\V_1\cdot \XCn_4^P) \)
		\tcp*{2M+1I}
		\(\V_3 \gets \V_2\cdot \XCn_4^P\)
		\tcp*{1M}
		\Return{\((\V_3\cdot \XCn_3, \V_3\cdot \XCn_2, \V_1\cdot\V_2)\)}
		\tcp*{3M}
    }
\end{algorithm}


\begin{algorithm}[ht!]
    \setstretch{1.1}
    \caption{{\tt xDBLADD}:
        combined pseudo-addition and doubling on \(\KChudnovsky\).
    }
    \label{alg:xDBLADD-kummer}
    \SetKwProg{Function}{function}{}{}
    \SetKwInOut{Cost}{Cost}
    \SetKwData{V}{V}
    \SetKwData{W}{W}
    \Function{{\tt xDBLADD}}{
        \KwIn{%
            \(\pm P, \pm Q\) in \(\KChudnovsky(\FF_p)\),
            and \((w_2,w_3,w_4) = {\tt WRAP}(\pm(P-Q))\) in \(\FF_p^3\) 
        }
        \KwOut{\( (\pm[2]P, \pm (P+Q)) \in \KChudnovsky(\FF_p)^2 \) }
        \Cost{%
            \(
            7\MUL +
            12\SQR + 
            12\MLC + 
            16\ADD +
            16\SUB
            \)
        }
        \SetKwData{V}{V}
        \( (\V_1,\V_2) \gets (\Hadamard(\V_1),\Hadamard(\V_2)) \)
        \;
        \( (\V_1,\V_2) \gets (\PPsqr(\V_1),\PPmul(\V_1,\V_2)) \)
        \;
        \( 
            (\V_1,\V_2) 
            \gets 
            \left(
                \PPmul(\V_1,
                    (\dualof{\epsilon}_1,\dualof{\epsilon}_2,\dualof{\epsilon}_3,\dualof{\epsilon}_4)
                ),
                \PPmul(\V_2,
                    (\dualof{\epsilon}_1,\dualof{\epsilon}_2,\dualof{\epsilon}_3,\dualof{\epsilon}_4)
                )
            \right)
        \)
        \;
        \( (\V_1,\V_2) \gets (\Hadamard(\V_1),\Hadamard(\V_2)) \) 
        \;
        \( (\V_1,\V_2) \gets (\PPsqr(\V_1),\PPsqr(\V_2)) \)
        \;
        \( (\V_1,\V_2) \gets
            \left(
                \PPmul(\V_1,
                    (\epsilon_1,\epsilon_2,\epsilon_3,\epsilon_4)
                    )
                ),
                \PPmul(\V_2, (1,w_2,w_3,w_4))
                )
            \right)
        \)
        \;
        \Return{\( (\V_1,\V_2) \)}
    }
\end{algorithm}

\subsection{Subroutines for signature verification}

The crucial {\tt Check} function for \(\KChudnovsky\)
(Algorithm~\ref{alg:kummer-check})
calls subroutines
{\tt BiiValues}
and 
{\tt BijValue} 
to compute the values of the biquadratic forms on \(\KInter\).
Algorithms~\ref{alg:bijvalue}
and~\ref{alg:biivalues}
are our simple implementations of these functions.
We choose to only store the four constants 
\(\squaredd_1\), \(\squaredd_2\), \(\squaredd_3\)
and \(\squaredd_4\),
but clearly one can gain some efficiency by
pre-computing more constants 
(\eg \(\squaredd_1\squaredd_2\), 
\(\squaredd_1\squaredd_4-\squaredd_2\squaredd_3\), etc.).
As the speed of this operation is not critical, 
it allows us to reduce
the number of necessary constants.
The four values of \(B_{11}\),
\(B_{22}\), \(B_{33}\), and \(B_{44}\)
are computed simultaneously,
since many of the intermediate operands 
are shared 
(as is clear from Equations~\eqref{eq:BInt_11}
through~\eqref{eq:BInt_44}).

\begin{algorithm}[ht]
    \setstretch{1.1} 
    \caption{        
        {\tt BijValue}:
        evaluates \emph{one}
        of the off-diagonal \(B_{ij}\)
        on~\(\KInter\).
    }
    \label{alg:bijvalue}
    \SetKwInOut{Cost}{Cost}
    \SetKwProg{Function}{function}{}{}
    \SetKwData{U}{U}
    \SetKwData{V}{V}
    \SetKwData{W}{W}
	\Function{{\tt BijValue}}{
        \KwIn{%
            \(\pm P\), \(\pm Q\) in \(\KInter(\FF_p)\)
            and
            \((i,j)\)
            such that
            \(\{i,j,k,l\}=\{1,2,3,4\}\)
        }
        \KwOut{%
            \(\BInt_{ij}(\pm P, \pm Q)\) in \(\FF_p\)
        }
        \Cost{%
            \(
            10\MUL +
            10\MLC + 
            1\ADD +
            5\SUB
            \)
        }
        \( 
            (
                \V_0, \V_1, \V_2, \V_3
            ) 
            \gets
            (
                Y_i^P \cdot Y_j^P,
                Y_k^P \cdot Y_l^P,
                Y_i^Q \cdot Y_j^Q,
                Y_k^Q \cdot Y_l^Q
            ) 
        \)
        \;
        \( 
            (
                \V_0,\V_2
            ) 
            \gets
            (
                \V_0 - \V_1,
                \V_2 - \V_3
            ) 
        \)
        \;
        \( 
            (
                \V_0,\V_1
            ) 
            \gets
            (
                \V_0 \cdot \V_2,
                \V_1 \cdot \V_3
            ) 
        \)
        \;
        \( 
            (
                \V_0,\V_1
            ) 
            \gets
            (
                \V_0 \cdot \squaredd_k\squaredd_l,
                \V_1 \cdot (\squaredd_i\squaredd_j-\squaredd_k\squaredd_l)
            ) 
        \)
        \;
        \( 
            \V_0
            \gets
            \V_0 + \V_1
        \)
        \;
        \( 
            \V_0
            \gets
            \V_0 \cdot \squaredd_i\squaredd_j(\squaredd_i\squaredd_k-\squaredd_j\squaredd_l)(\squaredd_i\squaredd_l-\squaredd_j\squaredd_k)
        \)
        \;
        \( 
            \V_0
            \gets
            \V_0 \cdot C
        \)
        \;
        \Return{\(\V_0\)}
	}
\end{algorithm}


\begin{algorithm}[ht]
    \setstretch{1.1} 
    \caption{        
        {\tt BiiValues}:
        evaluates \(B_{11}\), \(B_{22}\), \(B_{33}\), and \(B_{44}\)
        on \(\KInter\).
    }
    \label{alg:biivalues}
    \SetKwProg{Function}{function}{}{}
    \SetKwInOut{Cost}{Cost}
    \SetKwData{U}{U}
    \SetKwData{V}{V}
    \SetKwData{W}{W}
	\Function{{\tt BiiValues}}{
        \KwIn{%
            \(\pm P\), \(\pm Q\) in \(\KInter(\FF_p)\)
        }
        \KwOut{%
            \(
                ( \BInt_{ii}(\pm P, \pm Q) )_{i=1}^4
            \) in \(\FF_p^4\)
        }
        \Cost{%
            \(
            16\MUL +
            8\SQR +
            28\MLC + 
            24\ADD
            \)
        }
        \( 
            (
                \V, \W
            ) 
            \gets
            (
                \pm\Pc,
                \pm\Q
            ) 
        \)
        \;
        \( 
            (
                \V, \W
            ) 
            \gets
            (
                \PPsqr(\V),
                \PPsqr(\W)
            ) 
        \)
        \;
        \( 
            (
                \V, \W
            ) 
            \gets
            \left(
                \PPmul(\V,
                    (\dualof{\epsilon}_1,\dualof{\epsilon}_2,\dualof{\epsilon}_3,\dualof{\epsilon}_4)
                ),
                \PPmul(\W,
                    (\dualof{\epsilon}_1,\dualof{\epsilon}_2,\dualof{\epsilon}_3,\dualof{\epsilon}_4)
                )
            \right)
        \)
        \;
        \(
            \U
            \gets 
            \left(
                \begin{array}{c}
                    \Dotprod(\V, (\W_1,\W_2,\W_3,\W_4))\,,
                    \Dotprod(\V, (\W_2,\W_1,\W_4,\W_3))\,,
                    \\
                    \Dotprod(\V, (\W_3,\W_4,\W_1,\W_2))\,,
                    \Dotprod(\V, (\W_4,\W_3,\W_2,\W_1))
                \end{array}
            \right)
        \)
        \;
        \(
            \V
            \gets 
            \left(
                \begin{array}{c}
                    \Dotprod(\U, (\kappad_1,\kappad_2,\kappad_3,\kappad_4))\,,
                    \Dotprod(\U, (\kappad_2,\kappad_1,\kappad_4,\kappad_3))\,,
                    \\
                    \Dotprod(\U, (\kappad_3,\kappad_4,\kappad_1,\kappad_2))\,,
                    \Dotprod(\U, (\kappad_4,\kappad_3,\kappad_2,\kappad_1))
                \end{array}
            \right)
        \)
        \;
        \( 
            \V
            \gets
            \PPmul(\V, (\squaredd_1,\squaredd_2,\squaredd_3,\squaredd_4))
        \)
        \;
        \Return{\(\V\)}
	}
\end{algorithm}

\subsection{Subroutines for compression and decompression}

The compression and decompression functions
in Algorithms~\ref{alg:compress}
and~\ref{alg:decompress}
require the evaluation of the polynomials \(K_2\), \(K_3\), and~\(K_4\).
We used the simple strategy in
Algorithms~\ref{alg:getk2}, \ref{alg:getk3}, and~\ref{alg:getk4}
(\({\tt get\_K_2}\), \({\tt get\_K_3}\), and \({\tt get\_K_4}\), respectively),
which prioritises low stack usage over speed (which is again not critical here).

\begin{algorithm}[ht]
    \setstretch{1.1}
    \caption{        
        \({\tt get\_K_2}\):
        evaluates the polynomial \(K_2\) at \((l_1,l_2,\tau)\).
    }
    \label{alg:getk2}
    \SetKwProg{Function}{function}{}{}
    \SetKwInOut{Cost}{Cost}
    \SetKwData{V}{V}
    \SetKwData{W}{W}
	\Function{\({\tt get\_K_2}\)}{
        \KwIn{%
            \(\left(l_1,l_2,\tau\right)\) 
            with \(l_1,l_2\in \FF_p\)
            and \(\tau\in\left\{0,1\right\}\)
        }
        \KwOut{%
            \(K_2(l_1,l_2,\tau)\) in \(\FF_p\) as in Equation~(\ref{eq:k2})
        }
        \Cost{%
            \(
            1\MUL +
            3\SQR +
            6\MLC + 
            4\ADD +
            2\SUB
            \)
        }
        \begin{multicols}{2}
            \( \V \gets l_1\cdot q_2\)
            \;
            \( \V \gets l_2\cdot \V\)
            \;
            \If{\(\tau = 1\)}{
                \( \W \gets l_1\cdot q_0\)
                \;
                \( \V \gets \V + \W\)
                \;
                \( \W \gets l_2\cdot q_1\)
                \;
                \( \V \gets \V - \W\)
                \;
            }
            \( \V \gets \V \cdot q_3\)
            \;
            \( \V \gets \V + \V\)
            \;
            \( \W \gets l_1 + q_5\)
            \;
            \( \W \gets \W^2\)
            \;
            \( \V \gets \W - \V\)
            \;
            \( \W \gets l_2 \cdot q_3\)
            \;
            \( \W \gets \W^2\)
            \;
            \( \V \gets \W + \V\)
            \;
            \If{\(\tau = 1\)}{
                \( \W \gets q_4^2\)
                \;
                \( \V \gets \W + \V\)
                \;
            }
            \Return{\(\V\)}
        \end{multicols}
	}
\end{algorithm}

\begin{algorithm}[ht]
    \setstretch{1.1}
    \caption{        
        \({\tt get\_K_3}\):
        evaluates the polynomial \(K_3\) at \((l_1,l_2,\tau)\).
    }
    \label{alg:getk3}
    \SetKwProg{Function}{function}{}{}
    \SetKwInOut{Cost}{Cost}
    \SetKwData{U}{U}
    \SetKwData{V}{V}
    \SetKwData{W}{W}
	\Function{\({\tt get\_K_3}\)}{
        \KwIn{%
            \(\left(l_1,l_2,\tau\right)\) 
            with \(l_1,l_2\in \FF_p\)
            and \(\tau\in\left\{0,1\right\}\)
        }
        \KwOut{%
            \(K_3(l_1,l_2,\tau)\) in \(\FF_p\) as in Equation~(\ref{eq:k3})
        }
        \Cost{%
            \(
            3\MUL +
            2\SQR +
            6\MLC + 
            4\ADD +
            2\SUB
            \)
        }
        \begin{multicols}{2}
            \( \U \gets l_2^2 \)
            \;
            \( \V \gets l_1^2 \)
            \;
            \If{\(\tau = 1\)}{
                \( \W \gets \U + \V \)
                \;
                \( \W \gets \W \cdot q_2 \)
                \;
                \( \U \gets \U+1 \)
                \;
                \( \V \gets \V+1 \)
                \;
            }
            \( \U \gets \U\cdot l_1 \)
            \;
            \( \V \gets \V\cdot l_2 \)
            \;
            \( \U \gets \U\cdot q_1 \)
            \;
            \( \V \gets \V\cdot q_0 \)
            \;
            \( \V \gets \V - \U \)
            \;
            \If{\(\tau = 1\)}{
                \( \V \gets \V + \W \)
                \;
            }
            \( \V \gets \V \cdot q_3 \)
            \;
            \If{\(\tau = 1\)}{
                \( \U \gets l_1 \cdot l_2 \)
                \;
                \( \U \gets \U \cdot q_6 \)
                \;
                \( \U \gets \U \cdot q_7 \)
                \;
                \( \V \gets \V - \U \)
                \;
            }
            \Return{\(\V\)}
        \end{multicols}
	}
\end{algorithm}

\begin{algorithm}[ht]
    \setstretch{1.1}
    \caption{        
        \({\tt get\_K_4}\):
        evaluates the polynomial \(K_4\) at \((l_1,l_2,\tau)\).
    }
    \label{alg:getk4}
    \SetKwProg{Function}{function}{}{}
    \SetKwInOut{Cost}{Cost}
    \SetKwData{U}{U}
    \SetKwData{V}{V}
    \SetKwData{W}{W}
	\Function{\({\tt get\_K_4}\)}{
        \KwIn{%
            \(\left(l_1,l_2,\tau\right)\) 
            with \(l_1,l_2\in \FF_p\)
            and \(\tau\in\left\{0,1\right\}\)
        }
        \KwOut{%
            \(K_4(l_1,l_2,\tau)\) in \(\FF_p\) as in Equation~(\ref{eq:k4})
        }
        \Cost{%
            \(
            3\MUL +
            3\SQR +
            6\MLC + 
            4\ADD +
            2\SUB
            \)
        }
        \begin{multicols}{2}
        \If{\(\tau = 1\)}{
            \( \W \gets l_2 \cdot q_0 \)
            \;
            \( \V \gets l_1 \cdot q_1 \)
            \;
            \( \W \gets \W - \V \)
            \;
            \( \W \gets \W + q_2 \)
            \;
            \( \W \gets \W \cdot l_1 \)
            \;
            \( \W \gets \W \cdot l_2 \)
            \;
            \( \W \gets \W \cdot q_3 \)
            \;
            \( \W \gets \W + \W \)
            \;
            \( \V \gets l_1 \cdot q_3 \)
            \;
            \( \V \gets \V^2 \)
            \;
            \( \W \gets \V - \W \)
            \;
            \( \V \gets l_2 \cdot q_5 \)
            \;
            \( \V \gets \V^2 \)
            \;
            \( \W \gets \V + \W \)
            \;
        }
        \( \V \gets l_1 \cdot q_4 \)
        \;
        \( \V \gets \V \cdot l_2 \)
        \;
        \( \V \gets \V^2 \)
        \;
        \If{\(\tau = 1\)}{
            \( \V \gets \V + \W \)
            \;
        }
        \Return{\(\V\)}
        \end{multicols}
	}
\end{algorithm}